\documentclass{sig-alternate}
\usepackage{mdwlist}
\usepackage{algorithm}
\usepackage{algorithmic}
\usepackage{slashbox}
\usepackage{bm}

\newcommand{\UDoc}[2]{{U(#2|#1)}}       % Utility of document for intent t
\newcommand{\doc}{{d}}                    % Document

\begin{document}

\title{Structured Learning of Two-Level Dynamic Rankings}
\author{}

\newtheorem{mydef}{Definition}
\newtheorem{thrm}{Theorem}
\newtheorem{observ}{Observation}

\maketitle

\begin{abstract}
For ambiguous queries, conventional retrieval systems are bound by two conflicting goals. On the one hand, they should diversify and strive to present results for as many query intents as possible. On the other hand, they should provide depth for each intent by displaying more than a single result. Since both diversity and depth cannot be achieved simultaneously in the conventional {\em static} retrieval model, we propose a new {\em dynamic} ranking approach. Dynamic ranking models allow users to adapt the ranking through interaction, thus overcoming the constraints of presenting a one-size-fits-all static ranking.
In particular, we propose a new \textit{two-level} dynamic ranking model for presenting search results to the user. In this model, a user's interactions with the first-level ranking are used to infer this user's intent, so that second-level rankings can be inserted to provide more results relevant for this intent.
Unlike for previous dynamic ranking models, we provide an algorithm to efficiently compute dynamic rankings with provable approximation guarantees for a large family of performance measures. We also propose the first principled algorithm for learning dynamic ranking functions from training data.
In addition to the theoretical results, we provide empirical evidence demonstrating the gains in retrieval quality that our method achieves over conventional approaches.

\end{abstract}

% Diversified retrieval aims to tackle the problem of ambiguous queries, for which multiple user intents exist.
% However current approaches are unable to provide all users of such queries with sufficiently many results matching their interest.
% We believe the key to solve the problem is an under-utilized source of information in the retrieval process: the \emph{user}.
% We argue that as the user alone knows the true query intent, their feedback can be used to \emph{dynamically} determine the ranking they are shown.
% This feedback is incorporated into a novel \textit{two-level} model for presenting search results to the user, where user intent is {\em dynamically} inferred from the user's interaction with the first level.
% Using this inferred intent, second-level rankings are presented to the user, so as to provide them more results of their interest.
% We propose new submodular evaluation measures for diversified retrieval and provide an efficient way to optimize them for a two-level ranking.
% Using this optimization procedure, we pose the problem of predicting two-level rankings as a structural SVM problem.
% Finally we empirically demonstrate the efficacy of our method over conventional methods.

% 
% %%%%%%%%%%%%%%%%%%%%%%%
% 


% 
% %%%%%%%%%%%%%%%%%%%%%%%
% 

\section{Introduction}
\label{sec:intro}

%Search engine users tend to express different information needs using the same query, leading to the well-known problem of \emph{ambiguous queries}.
% Query ambiguity is common, where the ambiguity in intent can be very coarse (common examples include \emph{apple}, \emph{jaguar} and \emph{svm})
% At the same time these  distinctions could be more fine-grained say for the query \emph{apple ipod}, which could be used to seek information on buying the device or for reading reviews for the device.
 
Search engine users often express different information needs using the same query. This leads to the well-known problem of \emph{ambiguous queries}, where a single query can represent multiple intents. In some cases, the ambiguity in intent can be coarse (\emph{e.g.}, queries such as \emph{apple}, \emph{jaguar} and \emph{SVM}). In other cases, the distinctions can be more fine-grained (\emph{e.g.,} the query \emph{apple ipod} with the intent of either buying the device or reading reviews). 

%% <Q2K>: is SVM really that ambiguous a query?

%Standard retrieval methods (non-diversification methods) present results via maximizing the probability of relevance.
% However in practice this favors the most common user intent, thus in the process only satisfying users with that intent.
% On the other hand diversification-based methods \cite{carb:98, chen:06, yue:08, zhai:03}, place results from different user intents (\emph{i.e.} interpretations of the query) amongst the top few results.
% This may lead to some users being presented with atmost a single result of interest.
% The problem is that even a single result may not be enough to satisfy their information need.
% We define the \emph{depth} for a user session to indicate the number of relevant documents the user is provided.
% Thus in the process of covering all these different intents, there is a risk of not being able to provide sufficient depth for any of the users.
% Herein lies the underlying problem which current methods face: A trade-off between diversity and depth, where choosing one inevitably leads to a lack of the other.
% The question which arises is \emph{how can we get both diversity and depth?}

Conventional retrieval methods do not explicitly model query ambiguity, but simply select a ranking of results by maximizing the probability of relevance independently for each document \cite{Robertson/77}. A major limitation of this approach is that it favors results for the most prevalent intent. In the extreme, the retrieval system focuses entirely on the prevalent intent, but produces no relevant results for the less popular intents. Diversification-based methods (e.g. \cite{carb:98, zhai:03, chen:06, yue:08}) try to alleviate this problem by including at least one relevant result for as many intents as possible. However, this necessarily leads to fewer relevant results for each intent. Clearly, there is an inherent trade-off between \emph{depth} (number of results provided for an intent) and \emph{diversity} (number of intents served). In the conventional ranked-retrieval setting, choosing one invariably leads to the lack of the other. A natural question that arises in this context is: how can we obtain diversity while not compromising on depth?

%We believe the key to solving the problem of the aforementioned tradeoff, is better utilizing the feedback from the user itself.
% Although current methods leverage different information sources, they do not utilize the feedback users provides in response to the search results.
% Since the user has a clear idea of their intent, we believe their feedback is the best source of information to tackle the problem.
% This idea, first proposed in \cite{brandt:11}, was shown to \emph{theoretically} be able to tackle this issue.

%% <Q2K>: what are the different information sources? 

We argue that a key to solving the conflict between depth and diversity lies in the move from conventional {\em static} retrieval models to {\em dynamic} retrieval models \cite{brandt:11} that can take advantage of user interactions. Instead of presenting a single one-size-fits-all ranking, dynamic retrieval models allow users to adapt the ranking dynamically through interaction. Brandt et al. \cite{brandt:11} have already given theoretical and empirical evidence that even limited interactivity can greatly improve retrieval effectiveness. However, Brandt et al. \cite{brandt:11} did not provide an efficient algorithm for computing dynamic rankings with provable approximation guarantees, nor did they provide a principled algorithm for learning dynamic ranking functions from training data. In this paper, we resolve these two open questions.

%by better utilizing the user feedback.  Since users know their intent, their feedback is a valuable source of information. Brandt {\em et al.} \cite{brandt:11} theoretically showed that user feedback could be effectively incorporated in solving this problem. In this work, we propose a simpler yet more realistic model than that in \cite{brandt:11} for incorporating user feedback.

% In this work we propose a simple but more realistic model than that in \cite{brandt:11}, for incorporating user feedback.
% Our proposed \emph{two-level ranking} model utilizes the user feedback in order to determine further search results to be shown to  the user.
% The intent of users is inferred from their interaction (in the form of \emph{clicks} and \emph{skips}) with the first-level ranking.
% When users click on a document from the first-level ranking, they are presented a second level of results.
% The results presented at this level are related to the inferred intent of the user, and hence likely to match the user's interests, thus supplementing the user's information need.

In particular, we propose a new \emph{two-level dynamic ranking model}. The intuition behind the models is that the first level provides a diversified ranking of results. The system then senses the user's interactions with the first-level and interactively provides a second-level rankings conditioned on this feedback. A possible layout for such two-level rankings is given in Figure~\ref{fig:ex}. The left-hand panel shows the first-level ranking initially presented to the user. The user then chooses to expand the second document (\emph{e.g.} by mousing over or clicking) and a second-level ranking is inserted as shown in the center panel. Conceptually, the retrieval system maintains two levels of rankings as illustrated in the right-hand panel, where each second-level ranking is conditioned on the head document in the first-level ranking.

%of diverse set of documents is shown to a user based on her query. The intent of the user is inferred from her interaction 
%(in the form of \emph{clicks} and \emph{skips}) with the first-level ranking. When the user clicks on a document from the 
%first-level, she is presented with a second level of results. The results presented at the second level are related to the 
%intent of the user inferred from her interaction at the first level.  

\begin{figure*}[ht]
\begin{center}
\includegraphics[height=2.6in]{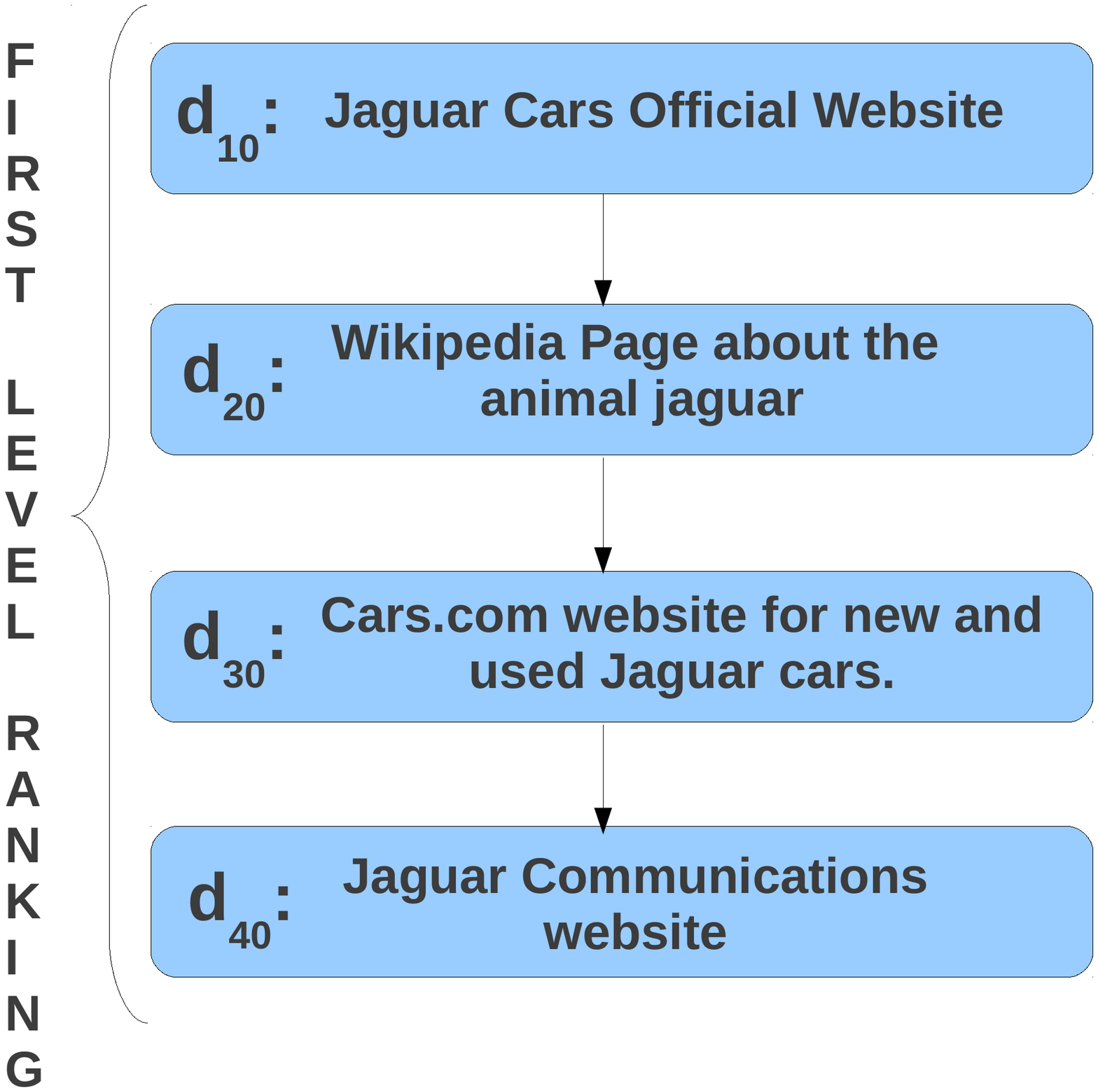}
\includegraphics[height=2.6in]{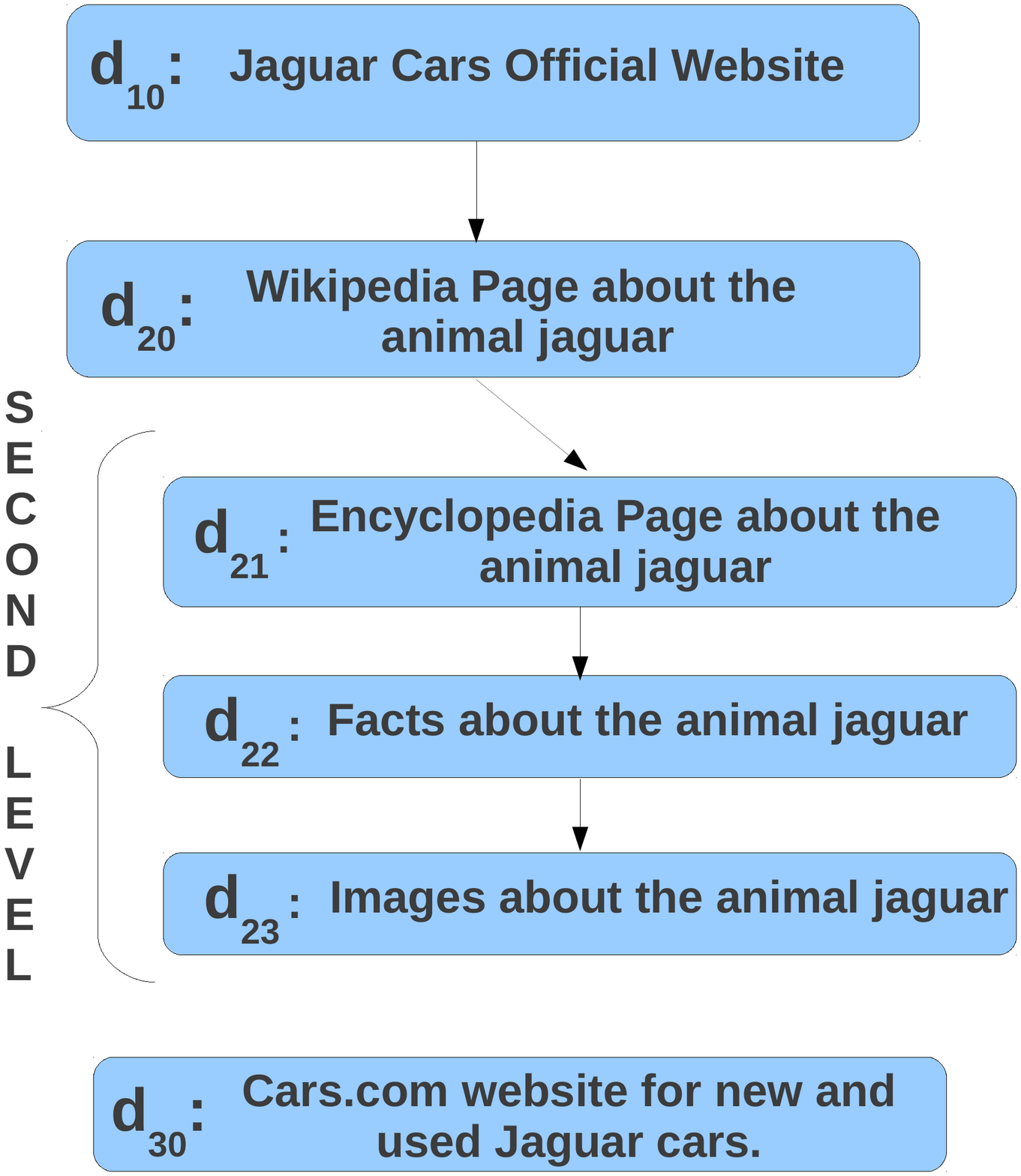}
\includegraphics[height=2.6in]{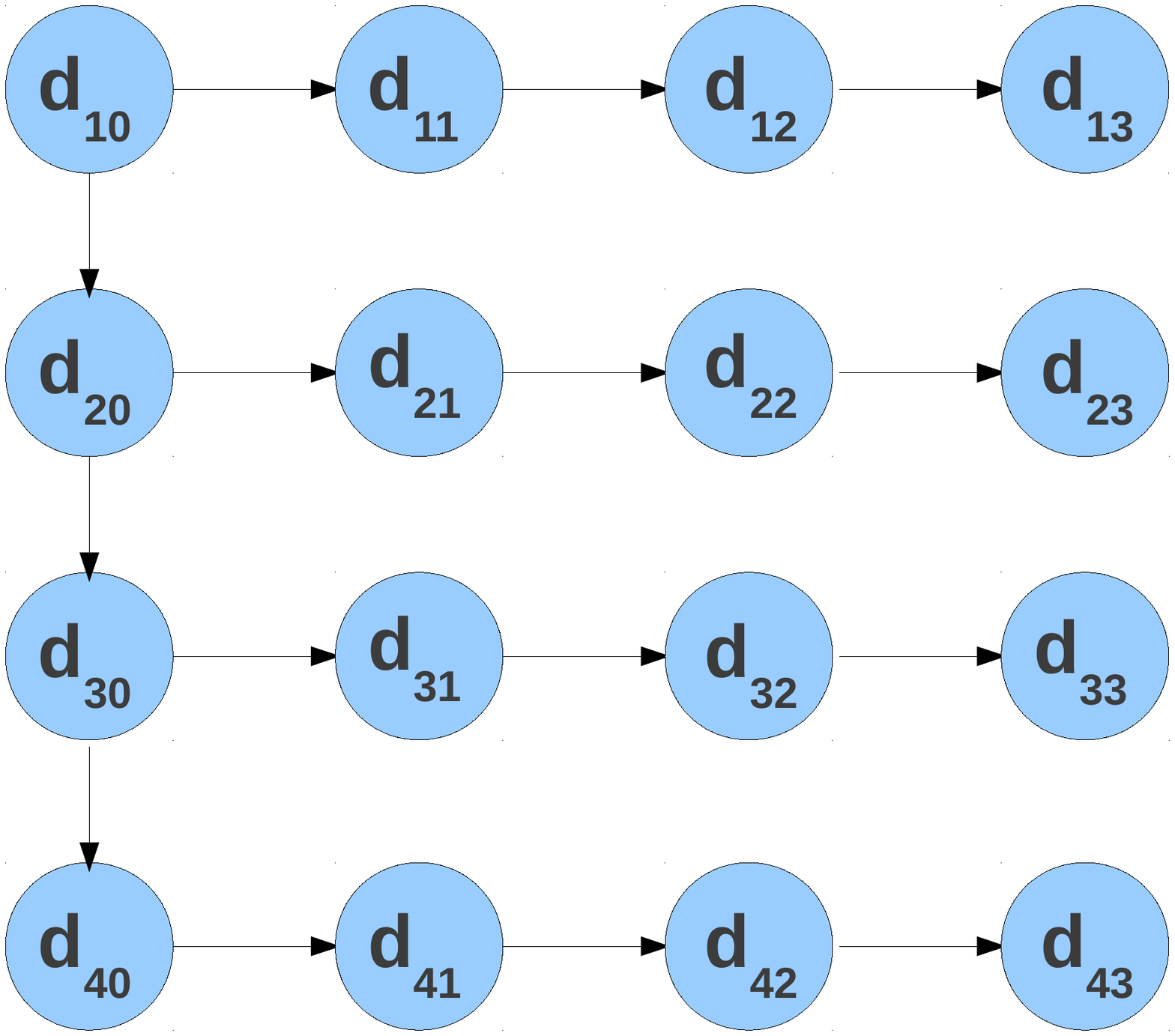}
\end{center}
\vskip -0.4in
\caption{A typical two-level ranking for the query ``jaguar''. A user interested in the animal ``jaguar'' interacts with the first-level ranking (left panel) and can expand results of interest to see additional results (middle panel). A  two-level rankings can be thought of as a two-dimensional matrix (right panel).}
\label{fig:ex}
\end{figure*}

% We also propose a new family of performance metrics specifically for diversified retrieval, which subsume most existing metrics.
% We treat the problem of compute a two-level ranking that optimizes these metrics as an optimization problem and demonstrate it to be an instance of a more general problem : a \emph{two-level} submodular optimization problem.
% We propose an simple but efficient greedy optimization procedure for this class of problems.
% Furthermore, we are able to prove an approximation bound for the proposed greedy algorithm as well.

To operationalize the construction and learning of two-level rankings in a rigorous way, we define a new family of performance measure for diversified retrieval. Many existing retrieval measures (e.g., Precision@k, DCG, Intent Coverage) are special cases of this family. We then operationalize the problem of computing an optimal two-level ranking as maximizing the given performance measure. While this optimization problem is NP-hard, we provide an algorithm that we show has a $(1 - \frac{1}{e^{1 - \frac{1}{e}}})$ approximation guarantee.

%These metrics subsume most of the existing metrics. We show that an effective two-level ranking can be constructed by optimzing the proposed metrics. This is in turn achieved by means of an efficient greedy two-level submodular optimization problem for which we provide an approximation guarantee. 

Finally, we also propose a new method for learning the (mutually dependent) relevance scores needed for two-level rankings. Following a structural SVM approach, we learn a discriminant model that resembles the desired performance measure in structure, but learns to approximate unknown intents based on query and document features. This method generalizes the learning method from \cite{yue:08} to two-level rankings and a large class of loss functions.
In addition to theoretical results, we evaluate the properties of our model, the algorithm for computing two-level rankings, and the learning methods through a detailed empirical analysis.

%To the best of our knowledge, we  are the first to provide a learning based framework for incorporating user feedback in ranking. 

% We are also the first to provide a principled learning-based framework for learning rankings that incoporate user feeedback.
% Our proposed methods uses the structural SVM learning framework to lean and predict two-level rankings.
% We use a \emph{discriminant} aimed at maximizing a submodular function of word coverage of a ranking.
% Thus maximizing the discriminant reduces to a two-level submodular optimization problem, which we solve using our proposed greedy optimization proceedure.
% We empirically demonstrate the benefits of our model over conventional diversified ranking methods and provide a detailed analysis of the results.

%Need to remove citation from Robertson-77

% 
% %%%%%%%%%%%%%%%%%%%%%%%
% 

\section{Related Work}
\label{sec:related}

%<Q2K> are there any good citations for vector-space models and language models with respect to learning to rank?

Traditional non-diversified methods for retrieval focus on ranking documents based on their probability of relevance to the query \cite{Robertson/77,lav:01,Liu:2009}. However, these models are problematic in the case of \emph{ambiguous queries}, as they tend to favor the most common user intent while ignoring the others.

Diversified retrieval aims to overcome the challenge of query ambiguity by providing diversity in search results \cite{carb:98, zhai:03, chen:06, yue:08, clarke:08}. In the extreme case diversified retrieval methods maximize {\it intent coverage}, meaning that they aim at covering as many intents in the ranking as possible by providing just a single document per intent. The methods in \cite{swam:08,rad:08,yue:08} formulate this problem as a set coverage problem. Most concretely, \cite{rad:08} proposed a multi-armed bandit algorithm, showing that it maximizes the number of users presented with at least one relevant document with provable guarantees. While diversification methods alleviated the problem of ignoring less frequent intents, they all explicitly or implicitly improve diversity at the expense of depth (i.e., they present only a few relevant documents for any given intent).

Recent work by Brandt et al. \cite{brandt:11} has focused on addressing the above issue.  They propose a \emph{dynamic} ranked-retrieval model that allows user interaction. User can interactively expand results so that a dynamic ranking gets created on the fly. Instead of following a static list of results, users construct their individual ranking as a path through a ranking tree. Since users with different intents take different paths, it is possible to tailor both the distribution and content of each path. Brandt et al. \cite{brandt:11} have shown that this small amount of interactivity overcomes the inherent trade-off between diversity and depth of conventional static rankings. We also follow this idea of dynamic ranking, but with the following differences. First, we focus on a simpler and more plausible model of user behavior. Unlike in \cite{brandt:11}, we do not assume that users are willing to provide feedback that is more than one level deep (see Section~\ref{sec:tldr}), and we allow users to backtrack to a higher level. Unlike the algorithms for constructing dynamic rankings presented in \cite{brandt:11}, we present an algorithm that has provable approximation guarantees. Furthermore, our algorithm and model apply to a large class of submodular performance measures, which include those of \cite{brandt:11} and \cite{rad:08,yue:08} as special cases. And finally, we propose a principled method for learning dynamic ranking functions. 

Our learning method follows a long line of research on training retrieval functions \cite{Liu:2009}. However, with the exception of \cite{yue:08}, virtually no other work directly addresses the problem of learning diversified retrieval functions. Following an idea first presented in \cite{swam:08}, Yue and Joachims \cite{yue:08} relate diversity in word occurences to diversity in intents. They propose to learn this relationship using a structural SVM method, where they formulate the discriminant function as a coverage problem and optimize intent coverage as the loss function.  We also employ a similar learning technique and also draw a correlation between intents and words. However, our learning method goes beyond single-level static ranking to predict two-level dynamic ranking, and it can optimize a large family of loss functions as defined in Section~\ref{sec:submod}.

The dynamic retrieval model we present in this paper is a special case of \emph{interactive retrieval}. An interactive retrieval setting involves multiple interactions between users and a system. Our model is most closely related to relevance feedback (e.g. \cite{aalb:92,Shen/Zhai/05a,Xu/Akella/08a,Zhang/Zhang/10a}), where the system presents a set of results and the user provides either implicit or explicit feedback. The feedback can then be used by the system to update the ranking. Note that the interface sketched in Figure~\ref{fig:ex} is inspired by SurfCanyon.com \cite{cramer:09}.
  
%Our work also relates to work on \emph{interactive retrieval}. In general, an interactive retrieval setting involves multiple user-system interactions. A simple example are relevance-feedback systems. Interactive retrieval systems driven by user-feedback, \emph{i.e.} user-driven systems, provide users with recommendations or allow users to provide more feedback. This feedback, can be explicit or implicit (like in the case of user clicks on search results). It can be used to provide better recommendations to the user \cite{cao:08} or in the case of relevance feedback, better search results \cite{aalb:92}. Our model can be considered a simple interactive-retrieval model, with the user feedback used helping increased the depth of rankings. In particular our work closely relates to work on Real-Time Implicit Relevance Feedback \cite{cramer:09}, as our method also incorporates user relevance-feedback in real-time, using a similar dynamic model.
  
%Last paragraph needs to be improved

% While our method is dynamic as in \cite{brandt:11}, we too draw a correlation between topics and words as in \cite{yue:08}, while learning rankings optimized on a new set of performance metrics.

% 
% %%%%%%%%%%%%%%%%%%%%%%%
% 

\section{Two-Level Dynamic Rankings}
\label{sec:tldr}

% What is dynamic
% What is static

 %Current methods for diversified retrieval are \emph{static} in nature.
 %A \emph{static} ranking is one that stays unchanged/static through a user session.
 %On the other hand, a \emph{dynamic} model of ranking can adapt the ranking presented, based on feedback received from the user.
 %The primary motivation of using a dynamic model is the inherent tradeoff between depth and diversity of static models.
 %To illustrate this tradeoff and how dynamic rankings can alleviate the problem, consider the following example:  

 Current methods for diversified retrieval are \emph{static} in nature.
 Such \emph{static} rankings stay unchanged through a user session.
 On the other hand, a \emph{dynamic} model can adapt the ranking based on interactions with the user.
 The primary motivation for using a dynamic model is the inherent trade-off between depth and diversity in static models. Figure \ref{fig:ex} illustrates the two-level dynamic rankings considered in this paper. We now provide a simple quantitative example to show how such two-level dynamic rankings can achieve both diversity and depth.
 
 \begin{table}
 \begin{center}
 \begin{tabular}{|c || c c c c c c c c c|}
 \hline 
 Intent &  $d_1$ & $d_2$ & $d_3$ & $d_4$ & $d_5$ & $d_6$ & $d_7$ & $d_8$ & $d_9$\\
 \hline 
 $t_1$ & 1 & 1 & 1 & 0 & 0 & 0 & 0 & 0 & 0\\
 $t_2$ & 0 & 0 & 0 & 1 & 1 & 1 & 0 & 0 & 0\\
 $t_3$ & 0 & 0 & 0 & 0 & 0 & 0 & 1 & 1 & 0\\
 $t_4$ & 0 & 0 & 0 & 0 & 0 & 0 & 1 & 0 & 1\\
 \hline 
 \end{tabular}
 \caption{ \label{tab:toyexample} Utility $U(d_j|t_i)$ of a document $d_j$ given an intent $t_i$.} 
 \end{center}
 \end{table}

 %Given 3 users and 9 documents with relevance judgments as given in Table \ref{table:toyexample} (with all users considered equally important).
 %A good \emph{static} non-diversified ranking method could present $d_1 \rightarrow d_2 \rightarrow d_3$ as its' top 3 documents.
 %While this achieve a high depth for $u_1$, it fails to provide diversity as users $u_2$ and $u_3$ are not presented with any documents relevant to them.
 %On the other hand, a good \emph{static} diversification system could present $d_1 \rightarrow d_4 \rightarrow d_7$ as the top 3 documents.
 %This successfully manages to achieve diversity, as it provides all users user with atleast 1 document, which is the goal of many diversification systems.
 %However it fails to provide any depth for any of the users.

 Consider four user intents $\{t_1,...,t_4\}$ and nine documents $\{d_1,...,d_9\}$ with relevance judgments $U(d_j|t_i)$ as given in Table \ref{tab:toyexample}. In this example, we assume that the user intents are equally likely. 
 On the one hand, a non-diversified static ranking method could present $d_7 \rightarrow d_8 \rightarrow d_9$ as its top three documents. This means that users with intents $t_3$ and $t_4$ get two relevant documents, but it fails to provide any relevant documents to users with intents $t_1$ and $t_2$. On the other hand, a diversified static ranking with $d_7 \rightarrow d_1 \rightarrow d_4$ as the top three documents covers all intents, but no user gets more than one relevant document. Therefore, this ranking lacks depth.

% Now instead consider a dynamic system, where users interested in $d_1$ (namely $u_1$), are also present $d_2$ and $d_3$ immediately.
% Similarly those users interested in $d_4$, would be presented with $d_5$ and $d_6$. 
% In this scenario starting with a ranking of $d_1 \rightarrow d_4 \rightarrow d_7$, while modifying the rankings within the search session to be dynamic as explained above, can lead to both diversity and high depth.
% More specifically, the above dynamic ranking achieves diversity as it provides each user with atleast 1 relevant document in the top 3 documents shown to the user.
% It also achieves higher depth as on average each user is presented with 2 relevant documents in the top 3 documents shown ($u_1$ with 3, $u_2$ with 2 and $u_3$ with 1).

As an alternative, now consider a two-level dynamic ranking as follows. The user is presented with $d_7 \rightarrow d_1 \rightarrow d_4$ as the first-level ranking. Users can now expand any of the first-level results to receive a second-level ranking. Assume that users interested in $d_7$ (and thus having intent $t_3$ or $t_4$) will \emph{expand} that result and receive a second-level ranking consisting of $d_8$ and $d_9$. Similarly, users interested in $d_1$ will get $d_2$ and $d_3$. And users interested in $d_4$ will get $d_5$ and $d_6$. Note that every intent is covered in the top three results of the first-level ranking. At the same time, users with intents $t_3$ and $t_4$ receive two relevant results in the top three positions of their dynamically constructed ranking $d_7 \rightarrow d_8 \rightarrow d_9 \rightarrow d_1 \rightarrow d_4$; users with intent $t_1$ also receive two relevant results in the top three positions of $d_7 \rightarrow d_1 \rightarrow d_2 \rightarrow d_3 \rightarrow d_4$; and users with intent $t_2$ still receive one relevant result in the top three of $d_7 \rightarrow d_1 \rightarrow d_4 \rightarrow d_5 \rightarrow d_6$. This example illustrates how a dynamic two-level ranking can provide diversity and increased depth simultaneously. 

% Thus identifying the intent of users interested in $d_1$ and \emph{dynamically} presenting them with $d_2$ and $d_3$ was the key to overcoming the tradeofff between depth and diversity.
% Taking this idea further, we propose a new \emph{two-level dynamic ranking} model for presenting search results to users.
% In our proposed model, the first-level of search results presented to a user resemble a ranking produced by a diversification system.
% However user interaction with these results, either by clicking or skipping results, provides the method with valuable feedback.
% This feedback is then used to infer the intent of the user, and provide them additional documents related to their intent in the form of a \emph{second-level} ranking.

In the above example, interactive feedback from the user was the key to achieving both depth and diversity. More generally, we assume that users interact with the dynamic ranking according to the following {\bf User Model}, which we denote as policy $\pi_d$. While other policies of user behavior are possible, we focus on $\pi_d$ for the sake of simplicity. Users expand a first-level document if and only if that document is relevant to their intent. When users skip a document, they continue with the next first-level result. When users expand a first-level result, they go through the second-level rankings before continuing from where they left off in the first-level ranking.  It is thus possible for a user to see multiple second-level rankings. Hence we do not allow documents to appear more than once across all two-level rankings.

 Note that this user model differs from the one proposed in \cite{brandt:11} in several way. First, it assumes only one level of feedback, while the model in \cite{brandt:11} assumes that users are willing and able to provide feedback many levels of rankings deep. Second, we model that users return to the top-level ranking, which is not allowed in the model of \cite{brandt:11}. We believe that these differences make the two-level model easier to understand for the user and therefore more plausible for practical use.

 %We define some notation used for the rest of this paper.
 %The initial set of documents shown on the first level, are called the {\bf head} documents, while the length of this first-level, denoted as $L$, is referred to as the {\bf length} of the two-level ranking.
 %Analogously, the set of documents on the second level are called the {\bf tail} documents.
 %A {\bf row} refers to a second-level ranking and its' corresponding head document.
 %The length of a row is denoted as $W$ and referred to as the width of the two-level ranking.
 %Static rankings are denoted as $\theta$ while two-level rankings are denoted as $\Theta$.
 %$\Theta_i$ refers to the $i^{th}$ row of a two-level ranking, with $\Theta_{i,0}$ representing the head document of the row and $\Theta_{i,j}$ denoting the $j^{th}$ tail document of the second-level ranking.

We now define some notation used in the rest of this paper. The set of documents shown on the first level are called the {\bf head} documents. The number of documents shown in this level is the {\bf length} of the first-level, and it is denoted by $L$. The set of documents shown on the second level are called the {\bf tail} documents. A {\bf row} denotes a particular {\bf head} document and all the {\bf tail} documents that follow it in the second level.  The length of a row, excluding the head document, is denoted by $W$ and is referred to as the {\bf width}. 
 Static rankings are denoted as $\theta$ while two-level rankings are denoted as $\Theta = (\Theta_1, \Theta_2, ... \Theta_i, ..)$.
 Here $\Theta_i = (d_{i0}, d_{i1}, ...., d_{ij}, ...)$ refers to the $i^{th}$ row of a two-level ranking, with $d_{i0}$ representing the head document of the row and $d_{ij}$ denoting the $j^{th}$ tail document of the second-level ranking. We denote the candidate set of documents to rank for a query $q$ by ${\cal D}(q)$, the set of possible intents by ${\cal T}(q)$ and $\mathbf{P}[t|q]~\forall t \in {\cal T}(q)$ denotes a distribution over the intents given a query $q$.  Unless otherwise mentioned, dynamic ranking refers to a two-level dynamic rankings in the rest of this paper.

% 
% %%%%%%%%%%%%%%%%%%%%%%%
% 

\section{Performance Measures for Diversified Retrieval}

%Structure in mind
%Para2: Current metric not enough. Need to be diverse.Then say that diminishing returns is observed in practice. Hence metric needs that as well.
%Para3: Define and explain the performance metric, introducing the required notation as well.

%  There have also been detailed user studies on understanding why users stop their search \cite{zach:05, brow:07, dost:09}.
%  These studies have found different reasons for users stopping, but the most common reason was that users found enough documents to fulfill their information need.

To define what constitutes a good two-level dynamic ranking, we start by defining the measure of retrieval performance we would like to optimize. We then design our retrieval algorithms and learning methods to maximize this measure. In the following, we first consider evaluation measures for one-level rankings, and then generalize them to the two-level case. 

\subsection{Measures for Static Rankings}

Existing performance measures range from those that do not explicitly consider multiple intents (e.g. NDCG, Avg Prec), to measures that reward diversity. Measures that reward diversity give lower marginal utility to a document, if the intents the document is relevant to are already well represented in the ranking. We call this the {\em diminishing returns} property. The extreme case is the ``intent coverage'' measure (e.g. \cite{swam:08,rad:08,yue:08}). It attributes utility only to the first document relevant for an intent and no additional utility to any additional documents.

%Most of the existing metrics typically account for either diversity or depth. For example, Radlinski {\em et al.} \cite{rad:08} provide an algorithm to give a single relevant document to each user (thus providing diversity but no depth). In contrast, the intent-aware measures proposed in \cite{agra:09} account for diversity at the expense of depth. In this section we propose a family of metrics that can account for both diversity and depth.  First, we provide a performance metric for single level ranking and then show that it can be naturally extended to a two-level ranking.

We now define a family of measures that includes a whole range of diminishing returns models, and that includes most existing retrieval measures. Let $g :\mathbb{R} \rightarrow \mathbb{R}$ with $g(0)=0$ be a concave, non-negative, and non-decreasing function that models the diminishing returns, then we define utility of the ranking $\theta=(d_1,d_2,...,d_k)$ conditioned on a given intent $t$ as
 \begin{equation}\label{eq:metric1}  U_{g}(\theta|t)=g\Big(\sum_{i=1}^{|\theta|}{\gamma_i \UDoc{t}{d_i}\Big)}. 
 \end{equation}
The $\gamma_1 \ge \gamma_2 \ge ... \ge \gamma_k \ge 0$ are discount factors that decrease with position in the ranking, and $\UDoc{t}{d}$ is the relevance rating of document $d$ for intent $t$.
The above definition of utility is with respect to a single user intent $t$. For a distribution of user intents $\mathbf{P}[t|q]$ for query $q$, we take the expectation
\begin{equation} \label{eq:metric}
  U_{g}(\theta|q)=\sum_{t\in {\cal T}(q)}{\mathbf{P}[t|q] \times U_{g}(\theta|t)}.
\end{equation}

%With the utility that we have defined, providing multiple documents that are relevant to the same intent has a diminishing return property. This can be easily seen since the function $g$ is concave. Thus, to maximize the above utility, a ranking should have both diversity and depth.

%<Q2K> what do you think about constructing a two intent example and showing how the metric is high when there is both diversity and depth?

Note that many existing retrieval measures are special cases of the definition above. For example, if one chose $g$ to be the identity function, one recovers the intent-aware measures proposed in \cite{agra:09} and the modular measures defined in \cite{brandt:11}. Further restricting $\mathbf{P}[t|q]$ to put all probability mass on a single intent leads to conventional measures like DCG \cite{jar:02} for appropriately chosen $\gamma_i$. At the other extreme, chosing $g(x)=\min(x,1)$ leads to the intent coverage measure \cite{rad:08,yue:08,swam:08} that assigns utility only to the first relevant document for an intent. Beyond these special cases, $g$ can be chosen from a large class of functions, implying that this family of performance measures covers a wide range of diminishing returns models.

\subsection{Measures for Dynamic Rankings}
\label{sec:submod}

 We are now ready to extend our family of performance measures to dynamic rankings. The key change for dynamic rankings is that users interactively adapt which results they view. 

How users expand first-level results was defined in Section~\ref{sec:tldr} as $\pi_d$. Under $\pi_d$, it is natural to define the utility of a dynamic ranking $\Theta$ as follows. For a user intent $t$ and a concave, non-negative, and non-decreasing function $g$, 
\begin{equation}\label{eq:util_t}
   U_{g}(\Theta|t) = g\bigg(\!\sum_{i=1}^{|\Theta|}\!{\Big(\!\gamma_i \UDoc{t}{\doc_{i0}}} + \!\sum_{j=1}^{|\Theta_i|}\gamma_{ij}{\UDoc{t}{\doc_{i0}} \UDoc{t}{\doc_{ij}}} \!\Big)\!\!\bigg).
 \end{equation}
Like for static rankings, $\gamma_1 \ge \gamma_2 \ge ...$ and $\gamma_{i1} \ge \gamma_{i2} \ge ...$ are position-dependent discount factors. Furthermore, we again take the expectation as in Equation \ref{eq:metric} to average over multiple user intents to obtain $U_g(\Theta|q)$.

Note that the utility of a second-level ranking for a given intent is zero unless the head document in the first-level ranking has non-zero relevance for that intent. This encourages second-level rankings to only contain documents relevant to the same intents as the head document, thus providing depths. The first-level ranking, on the other hand, provides diversity as controlled through the choice of function $g$. The ``steeper'' $g$ diminishes returns of additional relevant documents, the more diverse the first-level ranking gets. We will explore this in more detail in Section~\ref{sec:exconcave}.

%As before, adding mutiple documents relevant to the same intent at the first level has diminishing returns. Thus, diversity is encouraged at the first level.  Increased depth is encouraged at the second level due to the definition of the $\mathbb{I}_t(\Theta_{i,j}, \Theta_{i,0})$  which only rewards documents sharing an  intent with the head document.

A key advantage of modeling utility in the above form is that it allows for an efficient algorithm for finding a dynamic ranking which maximizes the utility. We present this in the next section.

\newcommand{\argmax}{\operatornamewithlimits{argmax}}

\section{Computing dynamic rankings}
\label{sec:algo}

In this section, we provide a greedy algorithm to compute dynamic rankings. These rankings are computed by maximizing the performance measures defined in the previous section. Computing rankings to exactly maximize our performance measure is an NP-hard problem. However, we show that our two-level greedy algorithm admits a $(1- e^{-1+1/e})$-approximation guarantee. 
%The greedy algorithm in this section is inspired by algorithms for several coverage \cite{coh:08, hoc:98, khu:99, submod} related problems which achieve an $(1-\frac{1}{e})$ approximation.

Our proposed greedy algorithm is given in  Algorithm \ref{alg1}. We are given a query $q$, a candidate set of documents ${\cal D}(q)$, intents ${\cal T}(q)$ with their distribution $\mathbf{P}[t|q]$, and a concave, non-negative and non-decreasing function $g$ that defines the utility in (\ref{eq:metric}). Our goal is to construct a two dimensional ranking of length $L$ and width $W$ to maximize the performance measure (\ref{eq:metric}). In the algorithm, the ``overloaded operator'' $\oplus$ denotes either adding a document to a row, or adding a row to an existing ranking.

The proposed algorithm works as follows. Every document in the remaining collection is considered as the head document of a candidate row. For each candidate row, $W$ documents are greedily added to maximize the utility of the resulting partial dynamic ranking. Once rows of length $W$ are constructed, the row which maximizes the utility is added to the ranking. The above steps are repeated until the ranking has $L$ rows.  

\begin{algorithm}
\caption{for computing a two-level dynamic ranking.}
\label{alg1}
\begin{algorithmic}
\STATE \textbf{Input:} $(q, {\cal D}(q), {\cal T}(q), \mathbf{P}[t|q]: t \in {\cal T}(q))$,  $g(\cdot)$, length $L$ and width $W$.
\STATE \textbf{Output:}  A dynamic ranking $\Theta$.
\STATE $\Theta \leftarrow \text{new\_two\_level}()$ 
\WHILE {$|\Theta| \le L$}
\STATE $bestU \leftarrow - \infty$
\FORALL {$ d \in {\cal D}(q)$ s.t. $d \notin \Theta$}
\STATE $row \leftarrow \text{new\_row}();~~~row.head \leftarrow d$
\FOR {$j=1$ to $W$} 
\STATE $bestDoc \leftarrow \argmax_{d' \notin \Theta \cup row} U_g(\Theta \oplus (row \oplus d')|q)$
\STATE $row \leftarrow row \oplus bestDoc$
\ENDFOR
\IF{$U_g(\Theta \oplus row|q) > bestU$}
 \STATE $bestU \leftarrow U_g(\Theta \oplus row|q);~~~bestRow \leftarrow row $
\ENDIF
\ENDFOR
\STATE $\Theta \leftarrow \Theta \oplus bestRow$
\ENDWHILE

\end{algorithmic}
\end{algorithm}

%\begin{algorithm}
%\caption{Finding a two-level ranking}
%\label{alg1}
%\begin{algorithmic}
%\STATE \textbf{Input}:- $T$: Topics along with their probabilities \\ \hskip 0.45in $D$: Docs. with their relevance labels
%\STATE \textbf{Output}:- $R$: A two-level ranking\newline
%\STATE \textbf{global} $count_t$: keeps track of the number of documents covering $t$ seen until now (as per user model)
%\STATE \textbf{global} $tcount_t$: Temporary version of the $count_t$ \newline
%\FOR{$i$=$1$ to \emph{L}}
%\FORALL{$d \in D$, s.t. $d \notin R$}
%\STATE $row \leftarrow new();$ $row[0] \leftarrow d$
%\STATE $tcount_t \leftarrow count_t$ //Reset temp. variable
%\STATE $rScore_{row} \leftarrow DocScore(d,d)$; Update  $tcount_t$
%\FOR{$j$=$1$ to \emph{W}}
%\STATE $best \leftarrow argmax_{d'}{(DocScore(d',d))}$: $d' \notin R \cup row$
%\STATE $row[j] \leftarrow best$
%\STATE $rScore_{row} += DocScore(best,d)$; Update  $tcount_t$
%\ENDFOR

%\ENDFOR
%\STATE $R_{i} \leftarrow argmax_{r}{(rScore_r)}$
%\STATE Update the $count_t$ variables
%\ENDFOR
%\STATE \textbf{return} $R$\\
%\vskip 0.1in

%\STATE \emph{Function} \textbf{DocScore($d$,$h$):}
%\STATE $score \leftarrow 0$
%\FORALL{$t \in T$, s.t. \emph{covers}($d,h,t$) = 1}
%\STATE $score \leftarrow score + $ $[f(tcount_t+1) - f(tcount_t)]$

%\ENDFOR
%\STATE \textbf{return} $score$\\
%\vskip 0.1in

%\end{algorithmic}
%\end{algorithm}

The proposed algorithm is extremely simple and efficient. The algorithm requires  $O(|{\cal T}|)$ space and $O(|{\cal T}||{\cal D}|^2)$ time, where $|{\cal T}|$ is the total number of intents and $|{\cal D}|$ is the number of candidate documents. The run time of the algorithm can be further improved using techniques such as lazy evaluation \cite{lazy}.

%\subsection{Approximation Bound}
%\label{subsec:approx}
We now derive an approximation bound for the greedy algorithm by relating it to the well-known 
problem of optimizing submodular set functions. First, recall the following definition of a 
submodular function.

 \begin{mydef}
 Given a set $U$, a function $f:2^{U}\rightarrow \mathbb{R}$ is said to be \textbf{submodular} {iff} for all elements $u \in U$ and all sets $X$ and $Y$, s.t. $X \subseteq Y \subseteq U$, we have 
\begin{align}
\label{eq:submodular}
f(X \cup \{u\}) - f(X) \geq f(Y \cup \{u\}) - f(Y).
\end{align}
 \end{mydef}

When a submodular function is monotonic (i.e., $f(Y) \le f(X)$, whenever $Y \subseteq X$) and normalized (i.e., $f(\phi) = 0$), greedily constructing a set of size $k$ gives an $(1-1/e)$ approximation \cite{submod} to the optimal.

Since the definition of our utility in (\ref{eq:metric}) involves a concave function, it is not hard to show that finding the next best row to add ({\bf outer step}) is a submodular maximization problem. Moreover, given the head document, finding the best row ({\bf inner step}) is also a submodular maximization problem. Thus, finding a dynamic ranking to maximize our utility is a \emph{nested} submodular maximization problem. Since submodular function maximization is a hard problem, we can only find an \emph{approximately} good row (rather than the best greedy row) to add in each step. In spite of this, we can show an approximation guarantee for the greedy two-level ranking algorithm. Our result generalizes submdoular function maximization from one level to two levels in the same spirit as Hochbaum and Pathria \cite{hoc:98} generalize the coverage problem from one level to two levels. 

 %Due to the use of the concave function in the utility, we can show that finding the best row of a ranking, given a head document is an instance of submodular maximization.
 %We can also similarly identify the problem of finding the best row as a submodular maximization.
 %Thus maximizing the utility amounts to maximizing this \emph{nested submodular problem}, where finding the best element at each step requires solving another submodular problem itself.
 %Here the two levels of the problem are:

%\begin{itemize}
%\item \textbf{Inner Step}: Given the previous rows, the current row and the head document of the current row, find the next document to add to the row.
%\item \textbf{Outer Step}: Given the previous rows, findthe next best row to add to the ranking.
%\end{itemize}
%We now give our approximation guarantee for {\bf Algorithm \ref{alg1}}.

\newtheorem{lemma}{Lemma}
\begin{lemma}
 The nested greedy algorithm for the nested submodular optimization problem has a   $(1- \frac{1}{e^{(1- \frac{1}{e})} })$ approximation bound.
\end{lemma}

\begin{proof}
 The submodular function in question is denoted $f$.
 We have $f$ normalized since the utility of the empty ranking is 0. Further, $f$ is 
 monotonic since the score can only increase on adding more rows.
 
 Let $S_i$ be the solution of the method after $i$ iterations of the outer step of the greedy algorithm.
 Let $OPT$ be the optimal solution to the problem with $k$ elements.
First, we define 
 \begin{equation} \label{eq:step1}
  \delta_i = f(S_i) - f(S_{i-1}).
 \end{equation}
 By monotonicity of the function $f$, we have: 
\begin{equation} \label{eq:step2}
\forall i, f(OPT) \leq f(S_i \cup OPT).                           
\end{equation}
 Since at every step we greedily pick the best element, by submodularity we get: 
\begin{equation} \label{eq:step3}
f(S_i \cup OPT) \leq f(S_i)+k \delta_{i+1}. 
\end{equation}
The above inequality follows from the fact that adding the elements of $OPT$ to the current solution has no more benefit than $k$ times the benefit achieved by the current best element.  

In the problem that we are considering, finding the best element (i.e., the inner step to get $S_i$ from $S_{i-1}$) is submodular as well. Hence, we are not assured of finding the best element to add to $S_{i-1}$; we can only obtain a $\beta$-approximate solution (where $\beta = 1 - \frac{1}{e}$). Thus we have $\delta_i = \beta \times \delta_i^{best}$. In this case the inequality (\ref{eq:step3}) becomes:
\begin{align} 
f(S_i \cup OPT) &\leq f(S_i)+k\delta_{i+1}/\beta,   \nonumber
\end{align}
which along with (\ref{eq:step2}) gives,
\begin{align}
 \delta_{i+1} &\geq \beta(f(OPT) - f(S_i))/k.  \nonumber
\end{align}
The above inequality in conjunction with (\ref{eq:step1}) implies,
\begin{align}
 f(S_{i+1}) &\geq f(S_{i}) + \beta(f(OPT) - f(S_i))/k \nonumber \\
  &= (1-\frac{\beta}{k})f(S_i) + \frac{\beta}{k} f(OPT). \label{eq:step7}
\end{align}
Using the above inequality, we can show by induction that $f(S_i) \geq (1 - (1- \frac{\beta}{k})^i)f(OPT)$. The base case with $i=1$ can be easily shown. For the induction step, using the inequality (\ref{eq:step7}) and the induction hypothesis we get:
\begin{align}
 f(S_{i+1}) &\geq (1-\frac{\beta}{k})(1 - (1- \frac{\beta}{k})^i)f(OPT) + \frac{\beta}{k} f(OPT) \nonumber\\
 &= (1- (1- \frac{\beta}{k})^{i+1})f(OPT).  \nonumber
\end{align}
Thus, after $k$ steps, the final solution $S$ satisfies $f(S) \geq (1- (1- \frac{\beta}{k})^{k})f(OPT)$
which implies $f(S) \ge (1 - e^{-\beta}) f(OPT)$. We get the required result by substituting the value of $\beta$. \end{proof}

% 
% %%%%%%%%%%%%%%%%%%%%%%%
% 

% % 
% % %%%%%%%%%%%%%%%%%%%%%%%
% % 

% 
% %%%%%%%%%%%%%%%%%%%%%%%
% 

\def\bw{{\bf w}}
\def\bpsi{ \Psi}

\section{Learning Dynamic Rankings}
\label{sec:ourapproach}
In the previous section, we showed that a dynamic ranking can be efficiently computed when all the intents and relevance judgments for a given query are known. Of course, in practice these are not available. In this section, we therefore propose a learning algorithm that can predict dynamic rankings on previously unseen queries. Following the approach of Yue and Joachims \cite{yue:08}, our algorithm makes use of word-level features to discriminatively learn a model of the intent distribution for new queries. In particular, given a training set of documents with known intents, our algorithm learns the weight vector of a linear discriminant function. This discriminant function can then be used in Algorithm~\ref{alg1} as a substitute for intents and relevance judgments in order to predict dynamic rankings on unseen queries.

We now describe our learning approach, which is based on structural SVMs \cite{tso:05}. Our goal here is to learn a mapping from a query $q$ to a dynamic ranking $\Theta$. We pose this as the problem of learning a weight vector $\bw \in \mathbb{R}^N$  from which we can make a prediction as follows:
\begin{align}
\label{eq:jfm}
h_{\bf w}(q) = \argmax_{\Theta}~\bw^\top \bpsi(q,\Theta).
\end{align}
In the above equation $\bpsi(q,\Theta) \in  \mathbb{R}^N$ is a joint feature-map between the candidate set of documents ${\cal D}(q)$ and ${\Theta}$ given a query $q$.\footnote{Strictly, the joint feature-map should be $\bpsi({\cal D}(q),\Theta|q)$ for a  given  query $q$. For brevity, we simply denote this as $\bpsi(q,\Theta)$.}

In the structural SVM framework, given a set of training examples, $(q^i, {\Theta^i})_{i=1}^n$, a discriminant function is obtained by minimizing the empirical risk $\frac{1}{n}\sum_{i=1}^n \Delta(\Theta^i, h_{\bw}(q^i))$ where  $\Delta$ is a loss function.

%\begin{align}
%\label{eq:erm}
%\frac{1}{n}\sum_{i=1}^n \Delta(\Theta^i, h_{\bw}(q^i)).
%\end{align}
%where $\Delta$ is a loss function.

The above equation requires the knowledge of $\Theta^i$ in order to compute empirical risk. However, in practice, we are not given $\Theta^i$ with the training documents. Assuming that we are given $(q^i, {\cal D}(q^i),{\cal T}(q^i), \mathbf{P}[t|q]: t \in {\cal T}(q^i))_{i=1}^n$, we first compute a dynamic ranking $\Theta^i$ that maximizes the utility $U_g$ (approximately) from  Algorithm~\ref{alg1}. These $\Theta^i$'s will be treated as the training examples in the rest of this section. 

A key aspect of structural SVMs is to appropriately define the joint-feature map for the problem at hand. For our problem, the joint-feature map in (\ref{eq:jfm}) is defined such that:
\begin{align}
\label{eq:discriminant}
 { \bf w}^\top \bpsi(q, \Theta)  
\:&:=& \!\!\!\!\!\!\!\!\!\sum_{v \in V_{{\cal D}(q)}} \!\!\!\!{\bf w}_v^\top \phi_{v} U_{g}(\Theta|v) \:+ \!\!\!\!\!\!\!\!\!\!\!\sum_{s  \in V_{{{\cal D}(q)} \times {{\cal D}(q)}}} \!\!\!\!\!\!\!\!\!\!\bw_{s}^\top \phi_{s}(\Theta) ,
\end{align}
where $V_{{\cal D}(q)}$ denotes an index set over the  words in the candidate set ${\cal D}(q)$. The vector $\phi_v$ denotes word-level features (for example, how often a word occurs in a document) for the word corresponding to index $v$. The utility $U_{g}(\Theta|v)$ is analogous to (\ref{eq:util_t}) but is now over the words in the vocabulary (rather than over intents). In particular, a document provides utility $U(d|v)$ for a word $v$, if that word occurs in the document. The word-level features are reminiscent of the features used in diverse subset prediction \cite{yue:08}. The key assumption is that the words in a document are correlated with the intent. This seems natural since documents relevant to the same intent are likely to share more words than documents that are relevant to different intents.

The second term in Eq.~\ref{eq:discriminant} captures the similarity between head and the tail documents. In this case, $V_{{\cal D}(q) \times {\cal D}(q)}$ denotes an index set over all document pairs in ${\cal D}(q)$. Consider a particular index $s$ that corresponds to documents $d_1$ and $d_2$ in ${\cal D}(q)$. $\phi_{s}(\Theta)$ is a feature vector describing the similarity between $d_1$ and $d_2$ in $\Theta$ when  $d_1$ is a head document in $\Theta$ and $d_2$ occurs in the same row as $d_1$. If either $d_1$ is not a head document in $\Theta$ or when $d_2$ is not in the same row as $d_1$,  $\phi_{s}(\Theta)$ is simply a vector of zeros.  An example of a feature in $\phi_{s}(\Theta)$ that captures the similarity between two documents is their TFIDF cosine. In the second term, the diminishing returns property does not hold strictly. However, it is easy to see that this term is modular ({\em i.e.,} Equation (\ref{eq:submodular}) holds with equality) and hence our greedy algorithm and its guarantee still hold even with these similarity features.    

From the above defintion of the feature-map (\ref{eq:discriminant}), it is clear that $ {\bf w}^\top \bpsi(q, \Theta)$ models the utility of a given dynamic ranking $\Theta$. Thus, a good discriminant function must give a higher value to rankings with higher utility. This is achieved by solving the following structural SVM optimization problem for ${\bf w}$ \cite{tso:05}:
\begin{align}
\label{eq:ssvm}
 \min_{{\bf w},\xi \ge 0}~&\frac{1}{2}||{\bf w}||^2 + \frac{C}{n} \sum_{i=1}^{n}{\xi_i}  \\
\text{s.t.}~ 
 & {\bf w}^\top  \bpsi(q^i,\Theta^i) - \bw^\top \bpsi(q^i,\Theta) \geq  \Delta(\Theta^i,\Theta|q^i) - \xi_i, \nonumber \\
  &~~~~~~~~~ \hfill \forall \Theta \neq \Theta^i,~~ \forall~ 1 \le i \le n \nonumber.
\end{align}
The constraints in the above formulation ensure that the predicted utility for the dynamic ranking $\Theta^i$ is
higher than the predicted utility for any other $\Theta$.  The objective function in (\ref{eq:ssvm}) minimizes the empirical risk while maximizing margin. 
The risk and the margin are traded off by the scalar parameter $C>0$. The loss between $\Theta^i$ and $\Theta$ is given by: 
\begin{align}
\Delta(\Theta^i,\Theta|q^i) := 1 - \frac{U_g(\Theta|q^i)}{U_{g}(\Theta^i|q^i)}. \nonumber
\end{align}
The above definition ensures that the loss is zero when $\Theta = \Theta^i$. It is easy to see that a dynamic ranking $\Theta$ has a 
large loss when its utility is low compared to the utility of $\Theta^i$. 

The quadratic program in Eq.~\ref{eq:ssvm} is convex and it can be solved efficiently using a cutting-plane algorithm \cite{joa:09,tso:05}. Even though Eq.~(\ref{eq:ssvm}) has an exponential number of constraints, the cutting-plane algorithm can be shown to always terminate in polynomial time \cite{joa:09,tso:05}. In each iteration of the cutting-plane algorithm, the most violated constraints in (\ref{eq:ssvm}) are added to a working set and the resulting quadratic program is solved.  Given a current ${\bf w}$, the most violated constraints are obtained by solving:
\begin{align}
\label{eq:most-violated}
\argmax_{\Theta}~{\bf w}^\top \bpsi(q^i,\Theta) + \Delta(\Theta^i,\Theta|q^i) 
\end{align}
It is easy to see that Algorithm~\ref{alg1} can be used to solve this problem, even thought the formal approximation guarantee does not hold in this case. While the original structural SVM was proposed for exact inference in Eq.~(\ref{eq:most-violated}), the approach has been shown effective \cite{yue:08,fin:08} even when only approximate inference is possible. Once a weight vector $\bw$ is obtained, the dynamic ranking for a test query can be obtained from Eq.~\eqref{eq:jfm}.

% 
% %%%%%%%%%%%%%%%%%%%%%%%
% 
% 

\section{Experiments}
\label{sec:experiments}
This section explores the properties of our two-level ranking method empirically. In particular, we first investigate how the choice of concave function $g$ impacts diversity and depth. We also compared against several static and dynamic baselines, and finally evaluate how accurately two-level dynamic rankings can be learned using the Structural SVM method.
 
All experiments were conducted on two datasets, namely,  the TREC 6-8 Interactive Track (TREC) and the Diversity Track of TREC 18 using the ClueWeb collection (WEB). Each query in TREC contains between 7 to 56 different manually judged intents. In the case of WEB, we used 28 queries with 4 or more intents.
%\footnote{Queries with fewer intents We do not consider queries with fewer than four intents as the problem becomes less interesting.}.  
Unless noted otherwise, we consider the probability $\mathbf{P}[t]$ of an intent proportional to the number of documents relevant to that intent. Key statistics describing the two datasets are provided in Table \ref{table:datasets}. Note that the two datasets differ vastly in terms of some criteria, therefore spanning a wide range of application scenarios. In particular, the most prevalent intent covers $73.4\%$ of all queries in the WEB dataset, while the most prevalent intent in TREC is far less dominating with $37.6\%$.

Unless noted otherwise, the number of documents in the first-level ranking is set to 5. The width of the second-level rankings is set to 2 (i.e. one head document plus 2 second-level results). For simplicity, we chose all factors $\gamma_i$ and $\gamma_{ij}$ in Equations~(\ref{eq:metric1}) and (\ref{eq:util_t}) to be $1$. Further, we chose $U(d|t)=1$ if document $d$ was relevant to intent $t$ and set $U(d|t)=0$ otherwise.

\begin{table}
\begin{center}
    \addtolength{\tabcolsep}{-0.5mm}
\begin{tabular}{ l | c | c}
Statistic &  TREC  &  WEB\\
\hline 
No. of queries & 17 & 28\\
Avg. \# of documents per query & 46.3 & 76.1\\
Avg. \# of intents per query  &  20.8 & 4.5\\
Avg. \# of docs with $>1$ intent per query &  9.6 & 25.6\\
Frac. of docs with $>1$ intent per query &  0.21 & 0.34\\
Avg. \# of intents per document &  1.33 & 1.41\\
Frac. of docs on prevalent intent &  0.376  & 0.734\\
%Importance of Prevalent Topic &  0.29  & 0.58\\
\end{tabular}
\caption{ \label{table:datasets} Key statistics of the two datasets.} 
\end{center}
\end{table}

\subsection{Controlling Diversity and Depth}
\label{sec:exconcave}
%No learning. Comparing optimal values.
%Comparing different concave functions to see role of diminishing returns in performance.
The key design choice of our family of utility measures is the concave function $g$. Since Algorithm \ref{alg1} directly optimizes utility, we first explore how different choices of $g$ affect various properties of the two-level rankings produced by our method.

%Our proposed family of performance measures is very large. In this experiment, we explore a range of measures covering the spectrum. The goal of this experiment is to study how the choice of concave function affects the dynamic ranking when the intents are known. 

% While our proposed family of performance measures is very large, we carefully examine some members of this family.
% We study which measures should be optimized to obtain \emph{good} dynamic rankings, along with analyzing the sensitivity of performance to this choice.
% By doing this optimization given the document-intent relevance-labels, we are comparing the potency of each method in a non-learning setting.
% We are particularly interested in understanding the impact of the diminishing return property of the concave function, used to obtain the performance measure as in Equation \ref{eq:util_t}.  

We experiment with four different concave functions $g$, each providing a different diminishing-returns model. At one extreme, we have the identity function $g(x)=x$ which corresponds to modular returns (i.e. Eq.~(\ref{eq:submodular}) holds with equality). Using this function in Eq.~(\ref{eq:metric1}) leads to the intent-aware Precision measure proposed in \cite{agra:09}, and it is the only function considered in \cite{brandt:11}. We therefore refer to this function as PREC. It is not hard to show that Algorithm \ref{alg1} actually computes the optimal two-level ranking for this choice of $g$.  On the other end of the spectrum, we study $g(x) = \min(x,2)$. By remaining constant after two, this function discourages presenting more than two relevant documents for any intent. The measure obtained using this function in Eq.~(\ref{eq:util_t}) will be referred to as SAT2 (short for ``satisfied after two'').  In between these two extremes, we study the square root function (SQRT) $g(x) = \sqrt{x}$ and the log function (LOG) $g(x) = \log(1+x)$.  A plot of all four functions is shown in Figure \ref{fig:concavecomp}.

 %We thus experiment with 4 such different concave functions, ranging from one extreme to another. The first function corresponds to the \emph{Identity} function, \emph{i.e} $g(x)=x$. Using the Identity function in Equation \ref{eq:util_t}, gives us the Expected Precision measure (represented as \emph{PREC}) introduced in \cite{brandt:11}. It is not hard to show that this can be optimally maximized, using the nested-greedy algorithm, to obtain a dynamic ranking. While the identity function is simple it has no diminishing returns as seen in Figure \ref{fig:concavecomp}. In between these two extremes, we study the square root function ($g(x) = \sqrt{x}$) and the  log function ($g(x) = \log(1+x)$).  

 %On the other hand we have the \emph{S2} function given by $g(x)=min(x,2)$. By remaining constant after 2, it discourages presenting more than 2 documents on an intent, and thus has rapidly diminishing returns. The measure obtained using S2 in Equation \ref{eq:util_t}, is denoted as \emph{SAT2} (short for \textbf{Satisfaction-2}).

\begin{figure}[t]
\begin{center}
\includegraphics[width=0.3\textwidth]{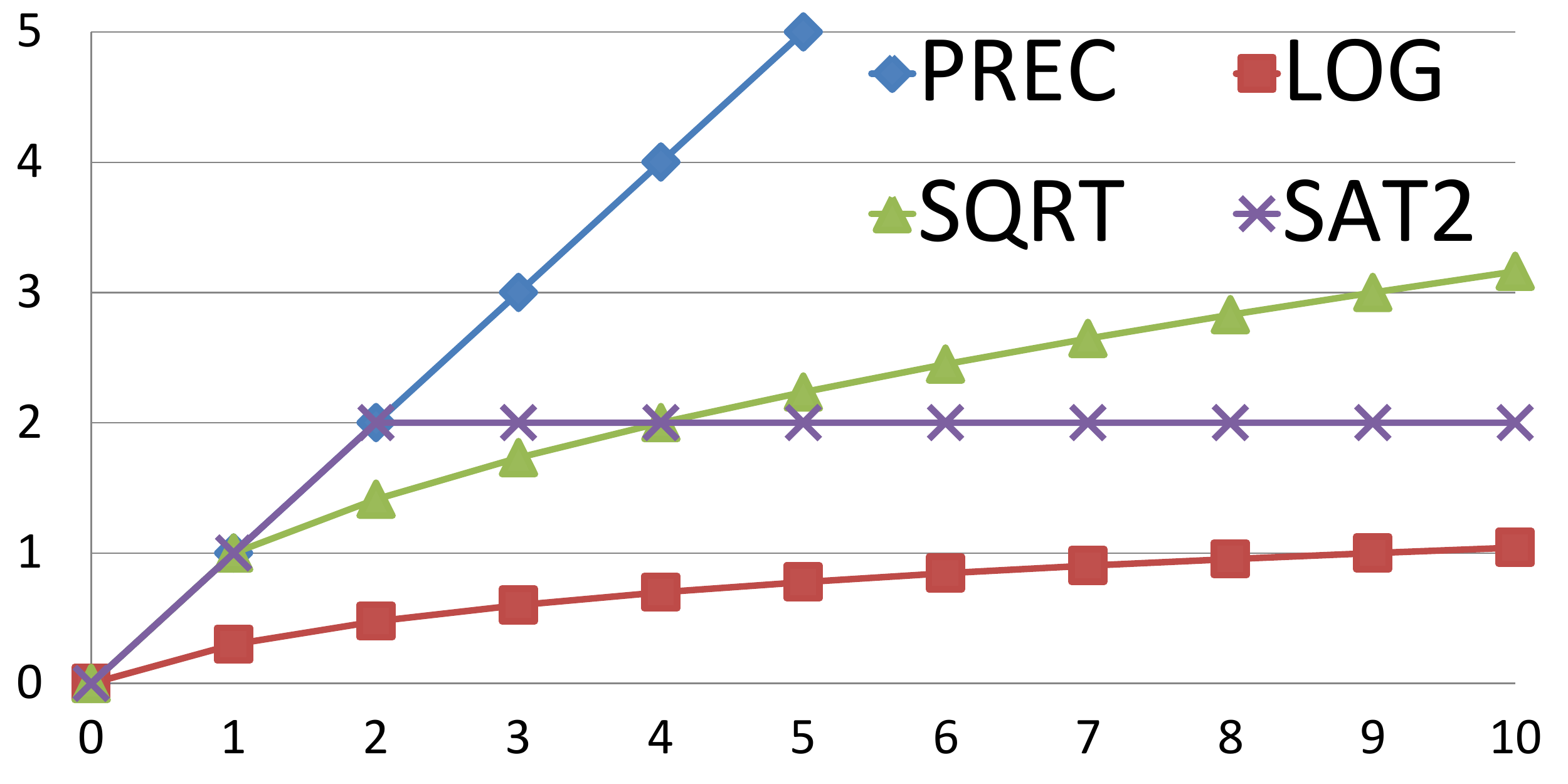}
\caption{Illustrating the diminishing-returns properties of four concave functions.}
%{\bf TJ: Make font and tic-marks larger. Use same names for functions consistently throughout paper.}
\label{fig:concavecomp}
\end{center}
\end{figure}
% \vskip -0.3in

 %In between these extremes we have many different concave functions. Amongst them, we chose the \emph{square root} function ($g(x)=\sqrt{x}$) and the \emph{log} function ($g(x)=\log{1+x}$) for their simplicity. They both lie in between these two extremes, with the log function looking more like S2 than thw square root function, as seen in the Figure. By incorporating them in the performance measure formulation of Equation \ref{eq:util_t}, we obtain two more performance measures: \emph{SQRT} and \emph{LOG} respectively.

%$({\cal D}, {\cal T}, \mathbf{P}[t])$, 

To explore how dynamic rankings differ for different choices of $g$, we used Algorithm~\ref{alg1} to compute the two-level rankings (approximately) maximizing the respective measure for known relevance judgments $\UDoc{t}{d}$ and $\mathbf{P}[t|q]$. Figure~\ref{fig:topiccoveragelevelone} shows how $g$ influences diversity.  The left-hand plot shows how many different intents are represented in the top 5 results of the first-level ranking on average. The graph shows that the stronger the diminishing-returns model, the more different intents are covered in the first-level ranking. In particular, the number of intents almost doubles on both dataset when moving from PREC to SAT2. In return, the number of documents on the most prevalent intent in the first-level ranking decreases, as shown in the right-hand plot. This illustrates how the choice of $g$ can be used to elegantly control the desired amount of diversity in the first-level ranking.

% {\bf ... now discuss figure. Also discuss depth for dominant intent.}

\begin{figure}[t]
\begin{center}
\includegraphics[width=0.235\textwidth]{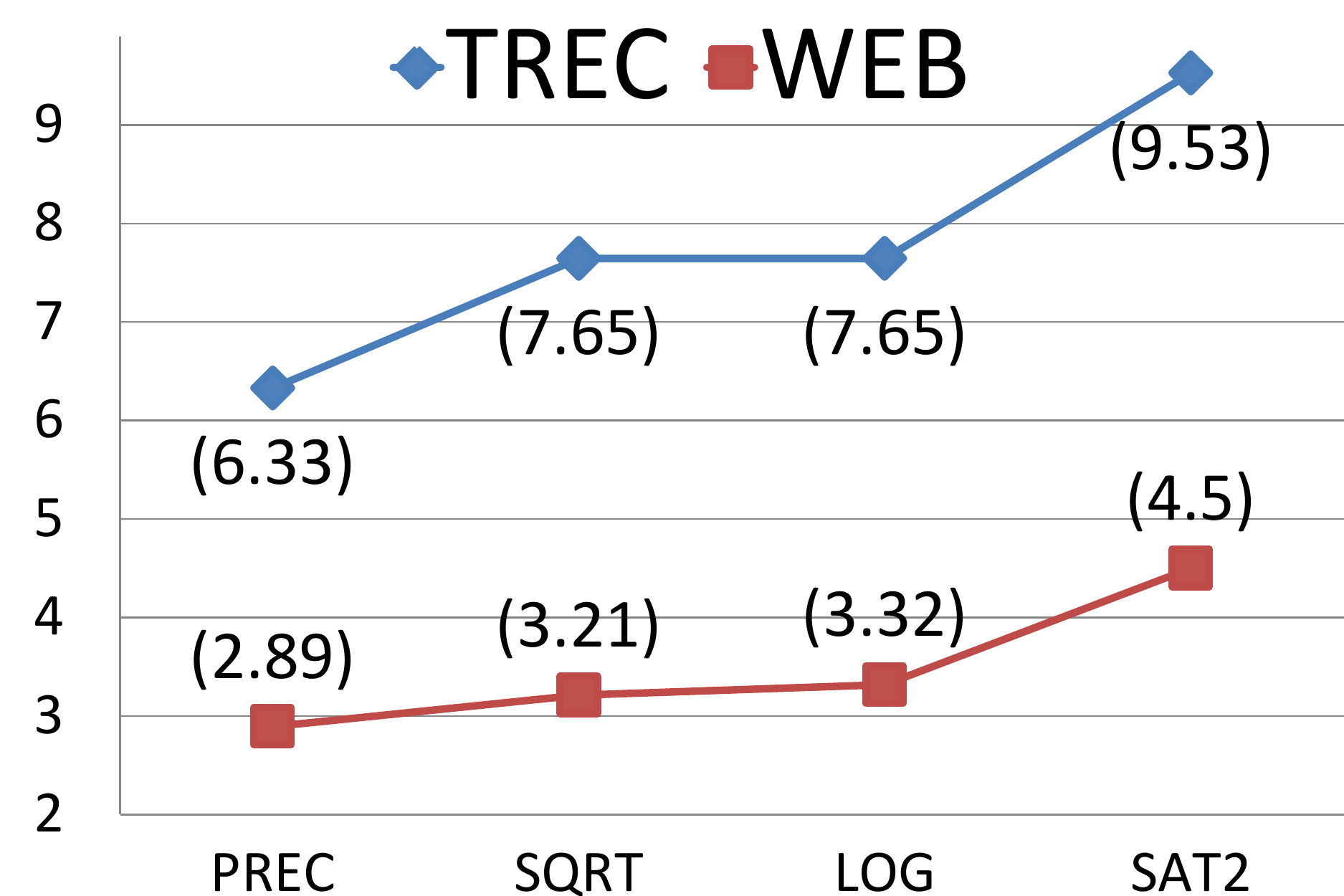}\hfill
\includegraphics[width=0.235\textwidth]{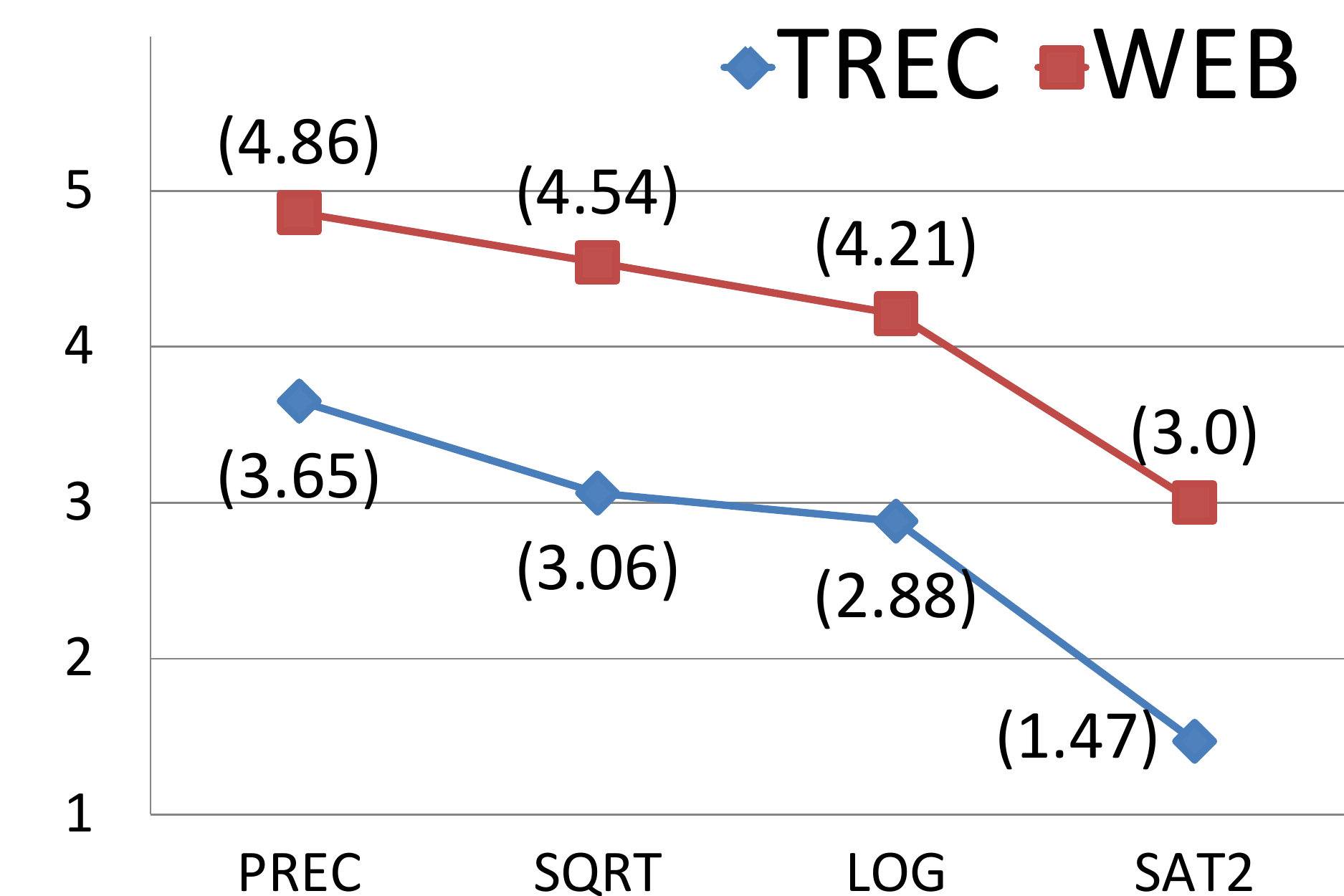}
\caption{Average number of intents covered (left) and average number of documents for prevalent intent (right) in the first-level ranking.}
\label{fig:topiccoveragelevelone}
\end{center}
\end{figure}

Tables \ref{tab:optvalues_Trec} (TREC) and \ref{tab:optvalues_web} (WEB) provide further insight into the impact of $g$, now also including the contributions of the second-level rankings. The rows correspond to different choices for $g$ when evaluating expected utility according to Eq.~(\ref{eq:util_t}), while the columns show which $g$ the two-level ranking was optimized for. Not surprisingly, the diagonal entries of Tables \ref{tab:optvalues_Trec} and \ref{tab:optvalues_web} show that the best performance for each measure is obtained when optimizing for it. The off-diagonal entries show that different $g$ used during optimization lead to substantially different rankings. This is particularly apparent when optimizing the two extreme performance measures PREC and SAT2; optimizing one invariably leads to rankings that have a low value of the other. In contrast, optimizing LOG or SQRT results in much smoother behavior across all measures, and both seem to provide a good compromise between depths (for the prevalent intent) and diversity. 

% In the rest of the experiments, we therefore use SQRT when optimizing rankings using Algorithm~\ref{alg1}.

%Optimizing SAT2 results in much lower utility with other measures. This is because of its extreme behavior where there is no additional benefit for more than two relevant documents. Similarly, optimizing PREC typically results in low SAT2. Comparing the two tables, we see that utilities are usually higher on the WEB dataset. We believe that this is primarily due to fewer intents and a much larger number of documents on the most prevalent intent. In spite of this differences, the overall trend is similar on both the datasets.

% We next computed the dynamic rankings that maximized each of the 4 described performance measures, using the nested greedy algorithm. We then compared the performance of each of these rankings on the same four measures. The results for the TREC and WEB datasets are given in Tables \ref{table:optvalues_trec} and \ref{table:optvalues_web} respectively. 
% Rows of the tables correspond to the values of a performance measure for different rankings and the columns indicate the measure being optimized for to produce the dynamic rankings.

\begin{table}[t]
\begin{center}
\begin{tabular}{ c || c c c c}
\hline
\backslashbox{Eval.}{Optim.} & PREC & SQRT & LOG & SAT2\\
\hline
PREC   & \textbf{0.315} & 0.302          & 0.294          & 0.164         \\
SQRT & 1.612          & \textbf{1.664} & 1.659          & 1.333         \\
LOG  & 1.216          & 1.267          & \textbf{1.27}  & 1.046         \\
SAT2 & 1.18           & 1.335          & 1.349          & \textbf{1.487}\\
\hline
\end{tabular}
\caption{Performance when optimizing and evaluating using different performance measures for TREC.}
\label{tab:optvalues_Trec}
\end{center}
\end{table}

\begin{table}[t]
\begin{center}
\begin{tabular}{ l || c c c c}
\hline
\backslashbox{Eval.}{Optim.} & PREC & SQRT & LOG & SAT2\\
\hline
PREC   & \textbf{0.746} & 0.731          & 0.714          & 0.443         \\
SQRT & 3.083          & \textbf{3.132} & 3.118          & 2.472         \\
LOG  & 2.236          & 2.297          & \textbf{2.303} & 1.908         \\
SAT2 & 1.773          & 1.882          & 1.892          & \textbf{1.984}\\
\hline
\end{tabular}
\caption{Same as Table \ref{tab:optvalues_Trec} for WEB.} 
\label{tab:optvalues_web}
\end{center}
\end{table}

\label{sec:results}

% 
% %%%%%%%%%%%%%%%%%%%%%%%
% 

\subsection{Static vs. Dynamic Ranking} \label{sec:statvsdyn}

The ability to simultaneously provide depth and diversity was a key motivation for our dynamic ranking approach over conventional static rankings. We now evaluate whether this goal is indeed achieved. We compare the two-level rankings produced by Algorithm~\ref{alg1} (denoted \emph{Dyn}) with several static baselines. 

First, we compare against a diversity-only static ranking, namely the static rankings obtained by maximizing intent coverage using the set-coverage algorithm proposed in \cite{yue:08} (denoted \emph{Stat-Div}). Second, we compare against a depth-only static ranking, namely the static ranking that optimizes utility with $g$ chosen to be the identity function (denoted \emph{Stat-Depth}). Note that Algorithm~\ref{alg1} can be used for this purpose by setting the width of the second-level rankings to $0$. And, third, we similarly use Algorithm~\ref{alg1} to produce static rankings that optimize SQRT, LOG, and SAT2 (denoted \emph{Stat-Util}). Note that both Dyn and Stat-Util optimize the same measure that is used for evaluation.

%Need to add the first one

\begin{figure}[t]
\begin{center}
\includegraphics[width=3.5in,height=1in ]{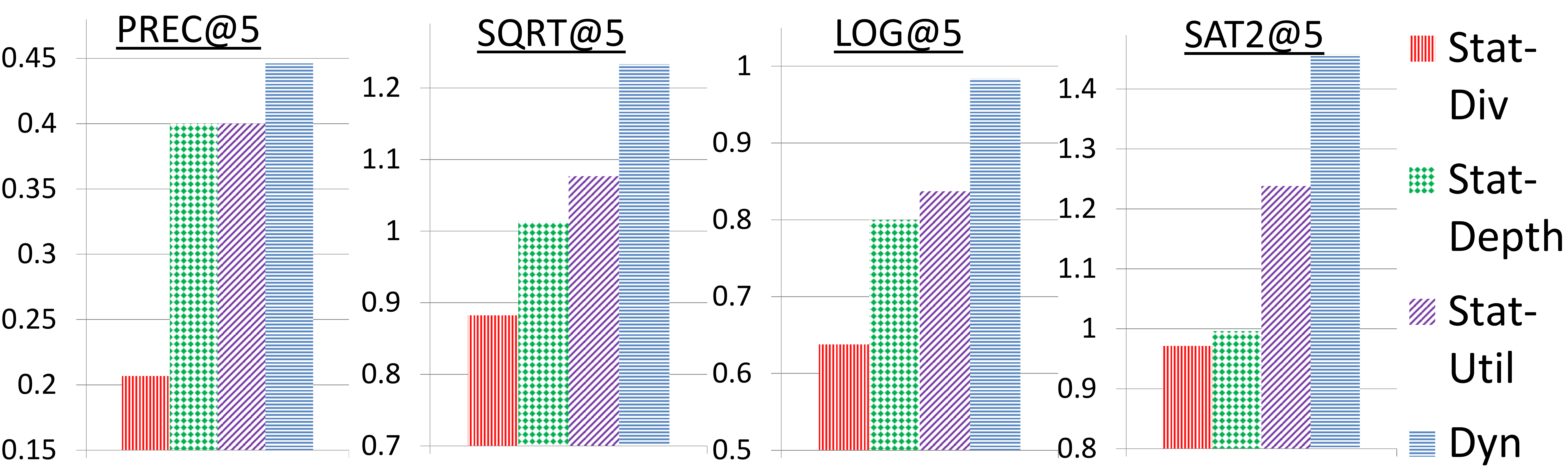}
\caption{Comparing the retrieval quality of Static vs. Dynamic Rankings for TREC.}
\label{fig:statdynopttrec}
\end{center}
\end{figure}

\begin{figure}[t]
\begin{center}
\includegraphics[width=3.5in,height=1in ]{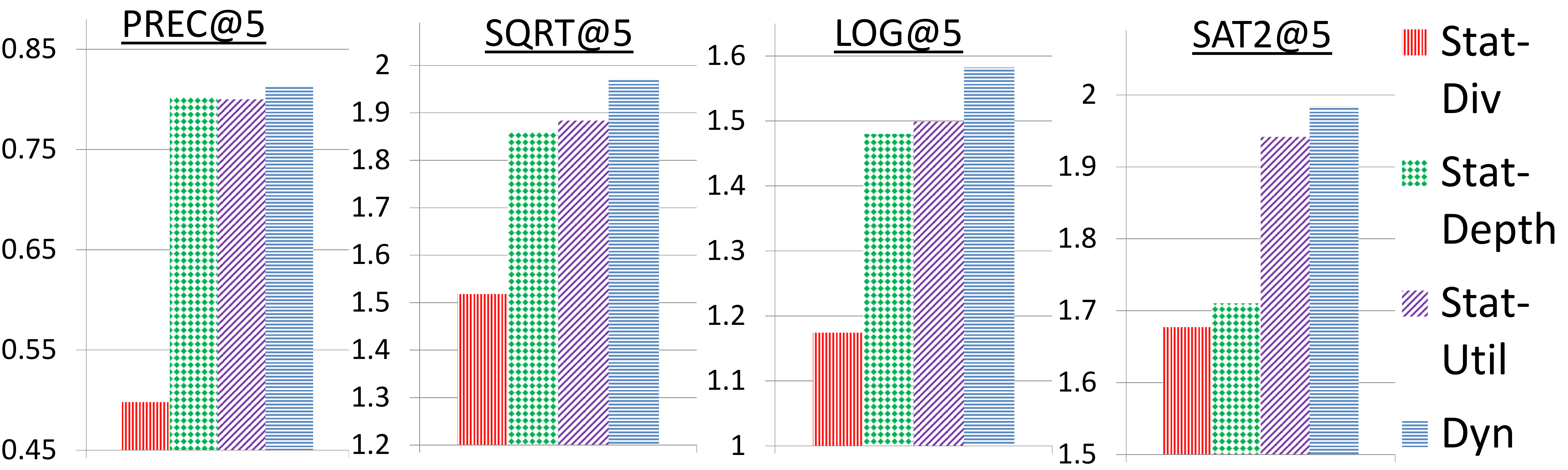}
\caption{Same as Figure \ref{fig:statdynopttrec} for WEB.}
\label{fig:statdynoptweb}
\end{center}
\end{figure}

% A key advantage our method has over that in \cite{yue:08}, is the ability to optimize rankings over a family of performance measure. In particular given a performance measure from this large family of metrics, our method allows us to find a dynamic ranking that explicitly optimizes this metric. However the method in \cite{yue:08} would require finding a ranking that optimizes intent coverage, and then compute the performance measure for that ranking. However we believe this is sub-optimal, in particular given the finding from our previous experiment, which showed dynamic rankings that optimize a measure, outperform others on that measure.
 
To make a fair comparison between static and dynamic rankings, we measure performance in the following way. For static rankings, we compute performance using the expectation of Eq.~(\ref{eq:metric1}) at a depth cutoff of $5$. In particular, we measure PREC@5, SQRT@5, LOG@5 and SAT2@5.
 For two-level rankings, the number of results viewed by a user depends on how many results he or she expands. So, we truncate any user's path through the two-level ranking after visiting $5$ results and compute PREC@5, SQRT@5, LOG@5 and SAT2@5 for the truncated path.

% We have also argued for the potency of dynamic rankings over static rankings. We intend to solidify this argument, by optimizing both static and dynamic rankings for the same performance measure, and then comparing on this performance measure. However a potential roadblock is the fact that the performance measures defined using Equation \ref{eq:util_t} were only for dynamic rankings, while those defined using Equation \ref{eq:metric} were only for static rankings. 

 %Defining the @5 metrics
% However there lies a simple solution to this problem. The document order presented to a user with a specific intent can be thought of as static rankings for that intent. In particular if we truncate the length of each document-ordering to the lenth of the static ranking being compared with \footnote{By setting the length of the dynamic and static rankings to be the same, say $k$, we ensure all orderings are atleast length $k$.}, allows us to directly compare dynamic rankings an static rankings.

% We thus take the performance measures discussed before, and use their truncated forms to make dynamic rankings comparable with the static rankings. Since we considered static and dynamic rankings of length 5, we suffix the previous measures with "@5" to denote they are the truncated forms, thus giving Prec@5, Log@5, Sqrt@5 and Sat2@5.

 Results of these comparisons are shown in Figures \ref{fig:statdynopttrec} and \ref{fig:statdynoptweb}. First, we see that both Dyn and Stat-Util outperform Stat-Div, illustrating that optimizing rankings for the desired evaluation measure leads to much better performance than using a proxy measure as in Stat-Div. Note that Stat-Div never tries to present more than one result for each intent, which explains the extremely low ``depth'' performance in terms of PREC@5. But Stat-Div is not competitive even for SAT2, since it never tries to provide a second result for the more prevalent intents. Second, at first glance it may be surprising that Dyn outperforms Stat-Depth even on PREC@5, despite the fact that Stat-Depth explicitly (and globally optimally) optimizes depth. As an explanation, consider the following situation where A is the prevalent intent, and there are three documents relevant to A and B and three relevant to A and C. Putting those sets of three documents into the first two rows of the dynamic ranking provides better PREC@5 than sequentially listing them in the optimal static ranking.

Overall, Figures~\ref{fig:statdynopttrec} and \ref{fig:statdynoptweb} show that the dynamic ranking method outperform all static ranking schemes on all the metrics -- in many cases with a substantial margin. This gain is more pronounced for TREC than for WEB. This can be explained by the fact that WEB queries are less ambiguous, since the single most prevalent intent accounts for more than 70\% of all queries on average.

% 
% %%%%%%%%%%%%%%%%%%%%%%%
% 

\subsection{Width of Second-Level Rankings} 

%TJ: Might make sense to include this in previous section
%Reply by KR: Agreed. Just wanted to get your optinion now that text is complete. 

 In the previous experiments the width of the second-level rankings was limited to $2$. To study the effect of width, we varied it from 0 (i.e. single-level, static) to 4. In each case we obtained a dynamic ranking optimized for the respective measure from Algorithm \ref{alg1}. We again use the truncated metrics as defined in Section~\ref{sec:statvsdyn} for evaluation. The results are shown in Figure~\ref{fig:nonlearnwidth}. Performance generally increases with increasing width on TREC. However, note that increasing width for SAT2 does not help much beyond width $1$, which is to be expected. The improvements from increased width are less strong on WEB, where not much gain is provided beyond width $1$. Again, this can be explained by the lower amount of query ambiguity.

\begin{figure}[!ht]
\begin{center}
\includegraphics[width=0.235\textwidth]{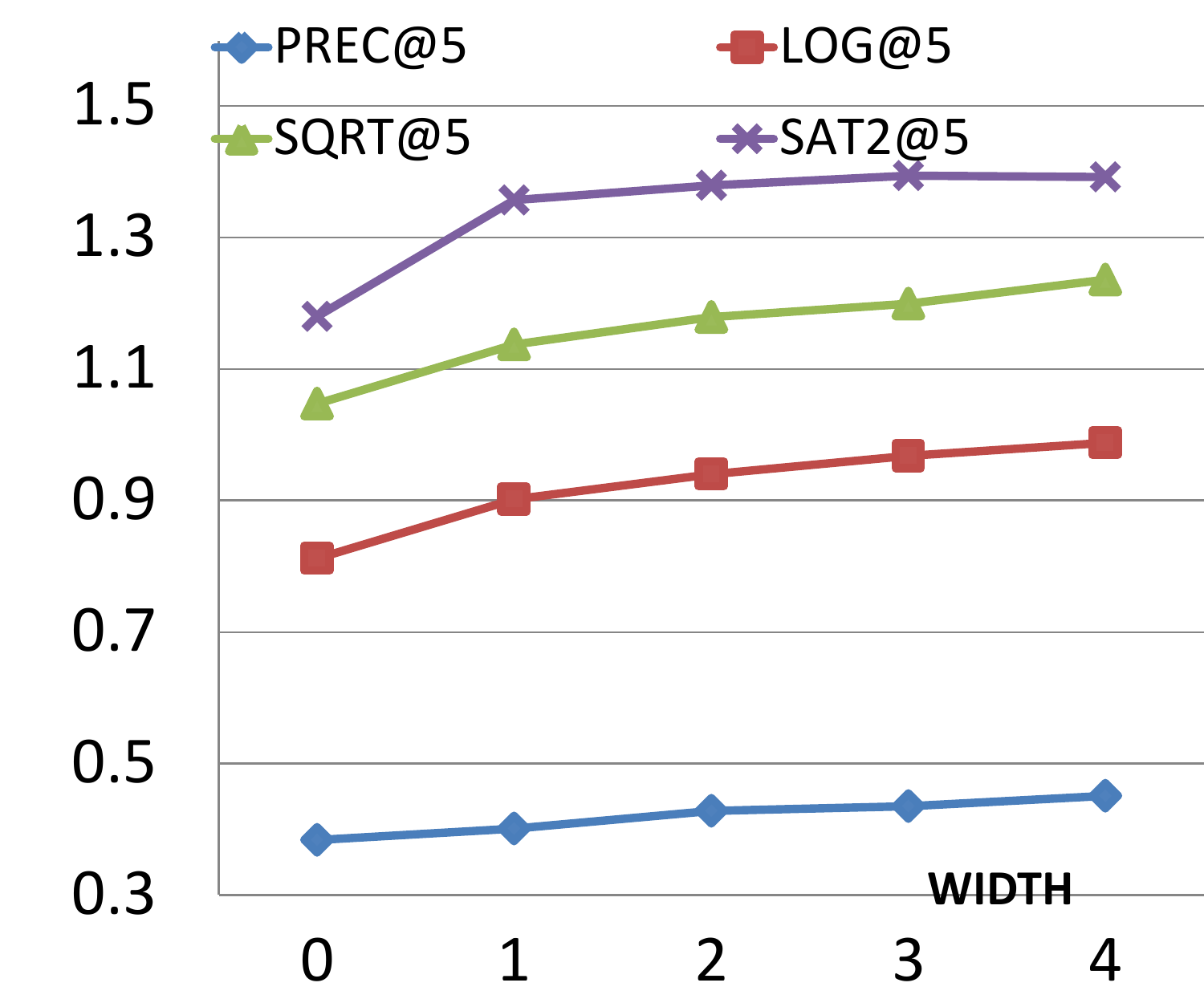}\hfill
\includegraphics[width=0.235\textwidth]{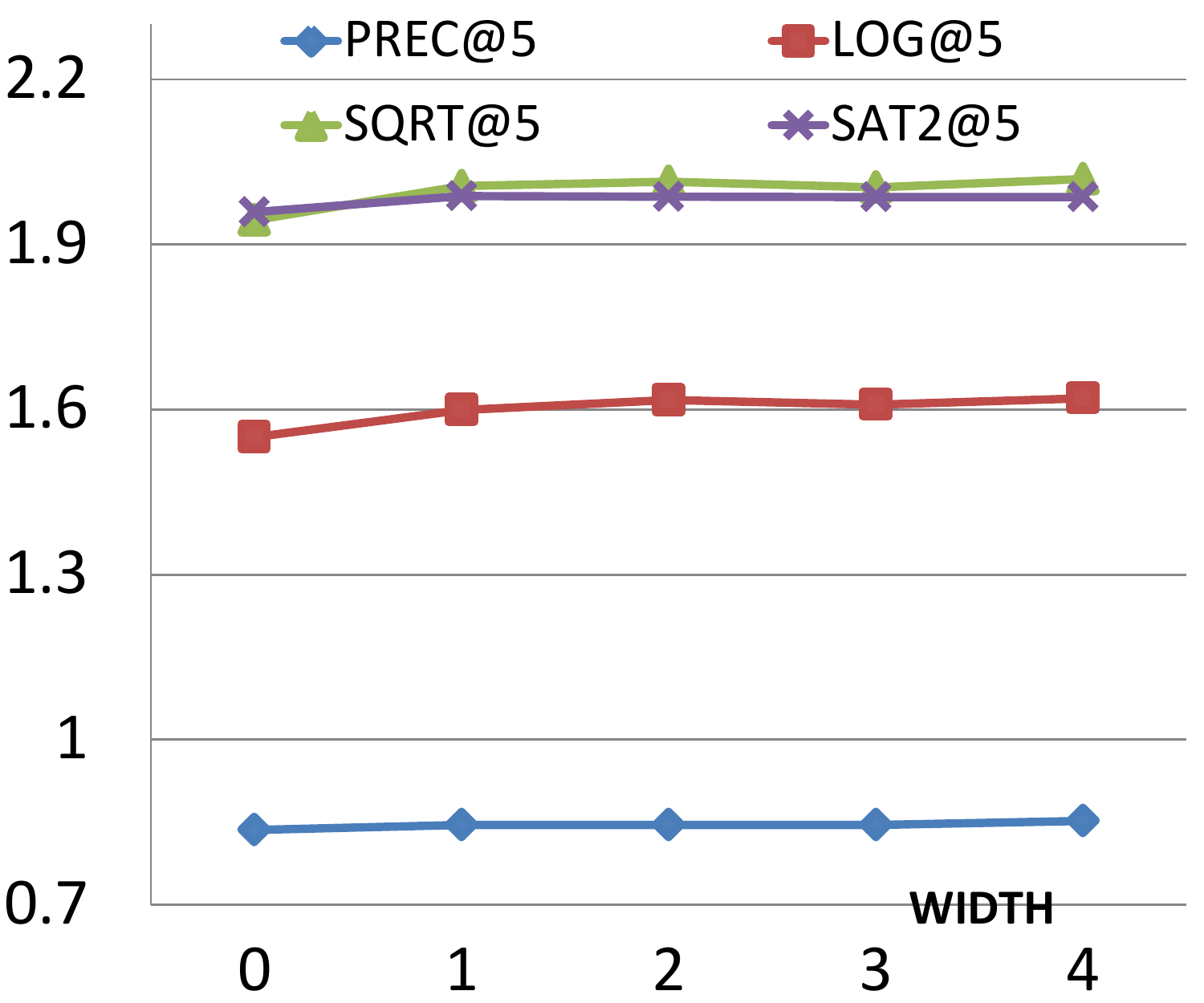}
\caption{Retrieval performance when the width of the second-level ranking is varied for TREC (left) and WEB (right).}
\label{fig:nonlearnwidth}
\end{center}
\end{figure}

% 
% %%%%%%%%%%%%%%%%%%%%%%%
% 

%\subsection{Comparison with Brandt. et. al.}
% Compare on Prec, Sqrt, Aban2

 %We also note that despite having a different model than that in \cite{brandt:11}, we are nearly able to match their \emph{best-case} performance across metrics, which indicates that the two-level ranking model is nearly as powerful as their model.  Further, we are able to provide a strong learning-method for the two-level model unlike them, hence making it more suited for use in practice.

% 
% %%%%%%%%%%%%%%%%%%%%%%%
% 

\subsection{Learning Two-level Ranking Functions}

So far we have evaluated  how far Algorithm~\ref{alg1} can construct effective two-level rankings if the relevance ratings are known. We now explore  how far our learning algorithm can predict two-level rankings for previously unseen queries.

For all experiments in this section, we learn and predict using SQRT as the choice for $g$, since it provides a good tradeoff between diversity and depth as shown above. We performed cross-validation as follows and report test-set performance averaged over all splits. 

For TREC, each test set consisted of a single held-out query. For each remaining set of 16 queries, 4 further splits were made such that 12 were used for training and 4 were used for validation. For WEB, we divided the data into 28 splits of 16 training, 8 validation and 4 testing. Queries were split such that all the queries were equally often in the training, test and validation sets respectively. The $C$ parameter of the structural SVM was varied from $10^{-5}$ to $10^{-1}$. The $C$ value corresponding to the best performance on the validation set was picked for each split. 

To compute features, we performed standard preprocessing such as tokenization, stopword removal and Porter stemming. Since the focus of our work is about diversity and not about relevance, we rank only those documents that are relevant to at least one intent of a query. This simulates a candidate set that may have been provided by a conventional retrieval method. This setup is similar to that used by Yue and Joachims \cite{yue:08}.

 Many of our features in $\phi_v$ follow those used in \cite{yue:08}. These features provide information about the importance of a word in terms of two different aspects. A first type of feature describes the overall importance of a word. These features capture, for example, the intuition that a word appearing in 10 documents in the candidate set is more important than a word appearing in only one.  Examples include:
\begin{itemize*}
 \item Word appears in at least $x\%$ of the documents?
 \item Word appears in the title of at least $x\%$ of the documents?
\end{itemize*}

A second type of feature describes the importance of a word in a particular document. This captures the intuition that a word that appears 10 times in a document is more important than a word that appears only once. Examples include:
\begin{itemize*}
 \item Word appears with frequency of at least $y\%$ within the document?
 \item Word that appears in $x\%$ of the documents, appears with frequency of at least $y\%$ within the document? %{\bf TJ: one other feature}
\end{itemize*}

Finally, we also use features $\phi_s$ that model the relationship between the documents in the second-level ranking and the corresponding head document of that row.  Examples of this type of feature include:
\begin{itemize*}
 \item TFIDF similarity of both documents, binned into multiple binary features,
 \item Number of common words that appear in both documents with frequency at least $x\%$.
\end{itemize*}

\begin{figure}[t]
\begin{center}
\includegraphics[width=3.5in,height=1in ]{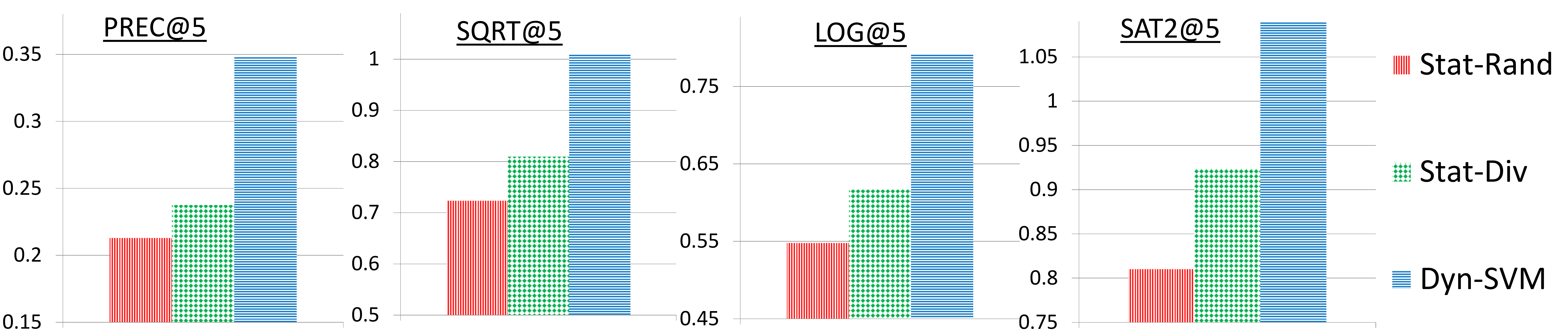}
\caption{Performance of learned retrieval functions, comparing static vs. dynamic rankings for TREC.}
\label{fig:statdyntrec}
\end{center}
\end{figure}

\begin{figure}[t]
\begin{center}
\includegraphics[width=3.5in,height=1in ]{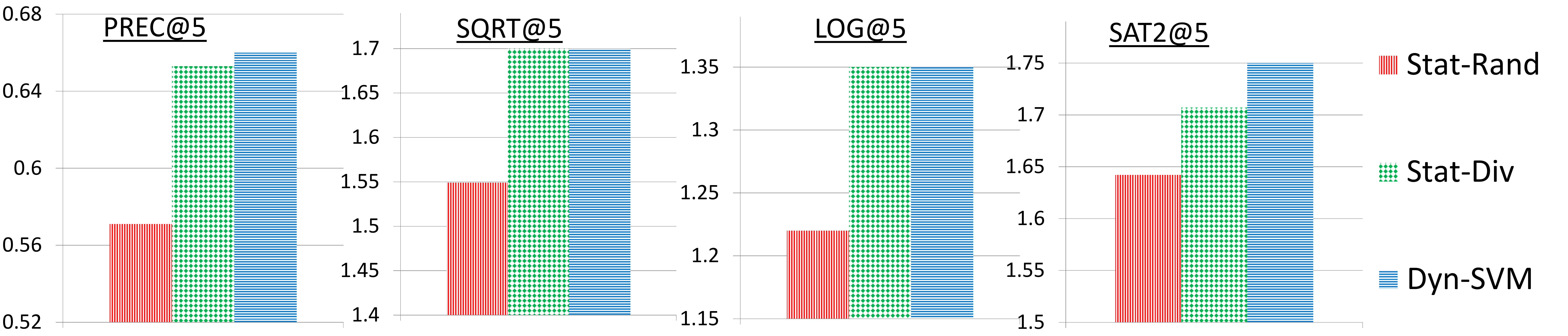}
\caption{Same as Figure \ref{fig:statdyntrec} for WEB.}
\label{fig:statdynweb}
\end{center}
\end{figure}

\paragraph{Dynamic vs. Static} In the first set of experiments, we compare our learning method (\emph{Dyn-SVM}) for two-level rankings with two static baseline. The first static baseline is the learning method from \cite{yue:08} which optimizes diversity, and is henceforth referred to as \emph{Stat-Div}.
 It is one of the very few learning methods for learning diversified retrieval functions, and it was shown to outperform non-learning methods like Essential Pages \cite{swam:08}.
 We also consider a \emph{random} static baseline (referred to as \emph{Stat-Rand}), which randomly orders the candidate documents.
 This is a competent baseline, since all our candidate documents are relevant for at least one intent.
 %Note that all results for the random baselines are obtained by averaging values over 100 trials for each query.
 %The third static baseline (referred to as \emph{Stat-Depth}) ranks highest all documents about the most prevalent intent, followed by documents on the next most common intent and so on. Although it is an unrealistic baseline, it can be considered as a \emph{skyline} for a depth-only ranker. TJ: This is not a learning method, so I dont know whether we really should use it here.
% Note that the dynamic baselines are quite strong, since TF-IDF accurately helps identify related documents on the same intent.

 Figure~\ref{fig:statdyntrec} shows the comparison between static and dynamic rankings for TREC. \emph{Dyn-SVM} substantially outperforms both static baselines across all performance metrics, mirroring the results we obtained in Section~\ref{sec:statvsdyn} where the relevance judgments were known. This shows that our learning method can effectively generalize the multi-intent relevance judgments to new queries. On the less ambiguous WEB dataset, Figure~\ref{fig:statdynweb} shows again that the differences between static and dynamic rankings are smaller. While \emph{Dyn-SVM} substantially outperforms \emph{Stat-Rand}, \emph{Stat-Div} is quite competitive on WEB.

\paragraph{Learning vs. Heuristic Baselines} We also compare against alternative methods for constructing two-level rankings. In particular, we extend the static baselines \emph{Stat-Rand} and \emph{Stat-Div} using the following heuristic. For each result in the static ranking, we add a second-level ranking using the documents with the highest TFIDF similarity from the candidate set. This results in two dynamic baselines, which we call \emph{Dyn-Rand} and \emph{Dyn-Div}.

\begin{figure}[t]
\begin{center}
\includegraphics[width=0.5\textwidth]{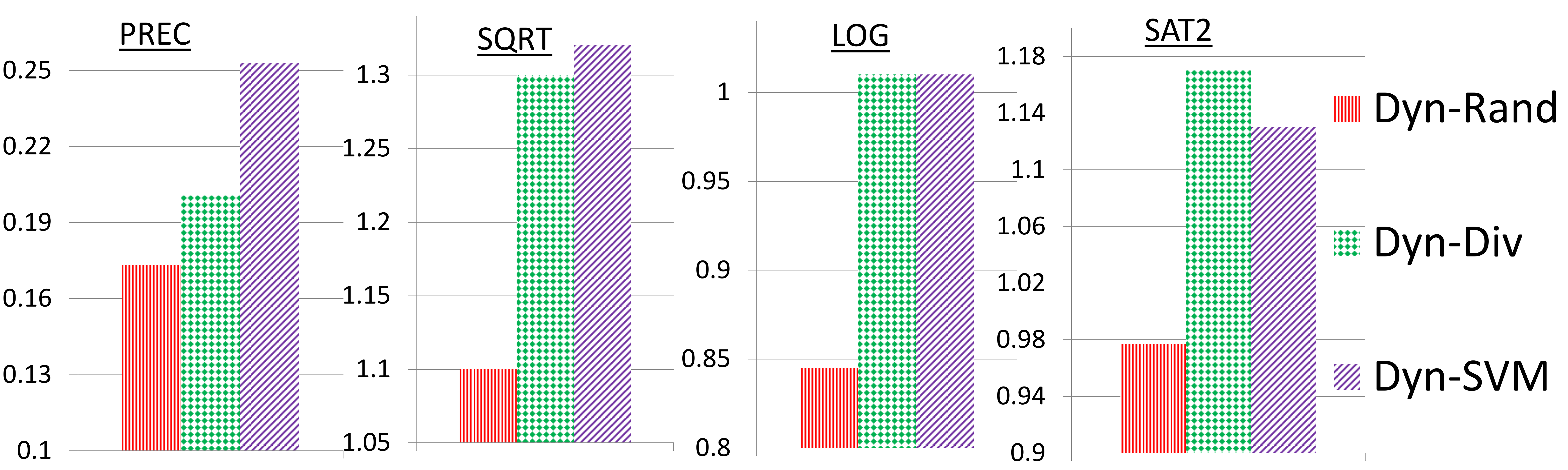}
\caption{Comparing learned dynamic rankings with heuristic baselines for TREC.}
\label{fig:dyndyntrec}
\end{center}
\end{figure}

\begin{figure}[t]
\begin{center}
\includegraphics[width=0.5\textwidth]{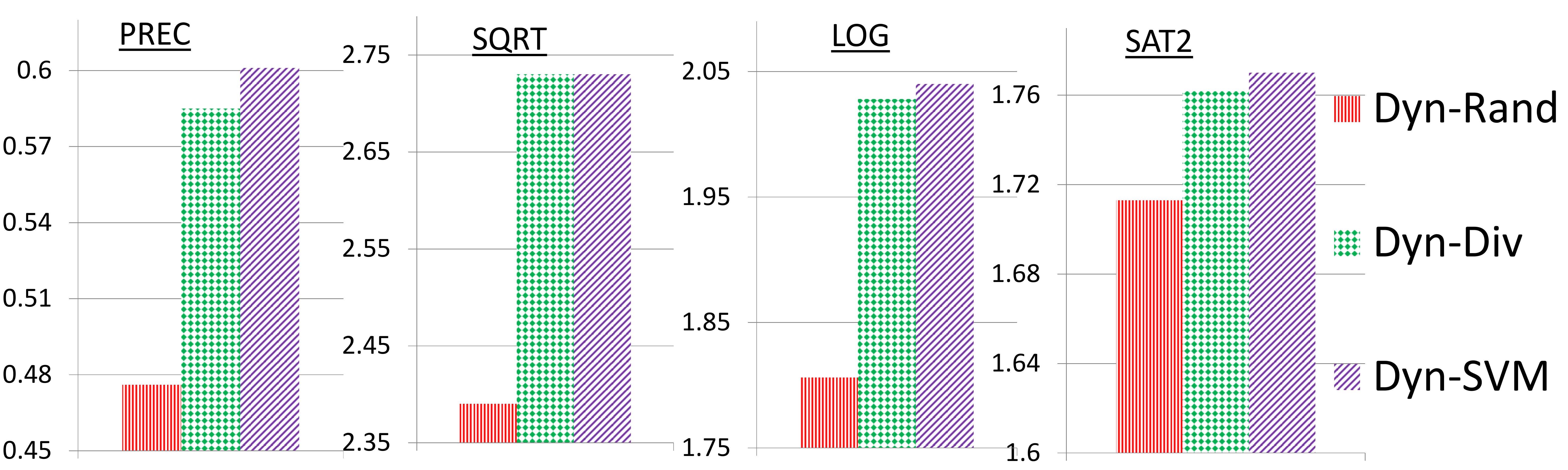}
\caption{Same as Figure \ref{fig:dyndyntrec} for WEB.}
\label{fig:dyndynweb}
\end{center}
\end{figure}

The results are shown in Figure~\ref{fig:dyndyntrec} for TREC and in Figure~\ref{fig:dyndynweb} for WEB. Since we are now comparing two-level rankings of equal size, we measure performance in terms of expected utility. On both datasets \emph{Dyn-SVM} performs substantially better than \emph{Dyn-Rand}. This implies that our method can effectively learn which documents to place at the top of the first-level ranking. Surprisingly, simply extending the diversified ranking of \emph{Dyn-Div} using the TFIDF heuristic produces dynamic rankings are are competitive with \emph{Dyn-SVM}. In retrospect, this is not too surprising for two reasons. First, our experiments with \emph{Dyn-SVM} use rather simple features to describe the relationship between the head document and the documents in the second-level ranking -- most of which are derived from their TFIDF similarity. Stronger features exploiting document ontologies or browsing patterns could easily be incorporated into the feature vector. Second, the learning method of \emph{Dyn-Div} is actually a special case of \emph{Dyn-SVM} when using the SAT1 loss (i.e. users are satisfied after a single relevant document) and second-level rankings of width 0. However, we argue that it is still highly preferable to directly optimize the desired loss function and two-level ranking using \emph{Dyn-SVM}, since the reliance on heuristics may fail on other datasets.

\section{Conclusions and Future Work}
\label{sec:conclusions}
 We proposed a two-level dynamic ranking approach that provides both diversity and depth for ambiguous queries by exploiting user interactivity. In particular, we showed that the approach has the following desirable properties. First, it covers a large family of performance measures, making it easy to select a diminishing returns model for the application setting at hand. Second, we presented an efficient algorithm for constructing two-level rankings that maximizes the given performance measure with provable approximation guarantees. Finally, we provided a structural SVM algorithm for learning two-level ranking functions, showing that it can effectively generalize to new queries.

The idea of dynamic ranking models that allow and actively anticipate user interactions, as well as the proof that such models can be learned and implemented with provable approximation guarantees, opens a wide range of further questions. First, we need user studies that investigate what types of user-interaction policies $\pi_d$ are most accurate in practice. While the two-level model used in this paper appears to be more plausible than the infinite-level model of \cite{brandt:11}, a detailed user study needs to investigate this question. Second, the learning approach presented in this paper requires relevance judgments for each intent of a query. While manually collecting such judgments is feasible in commercial settings like Web search engines, it would be desirable to have algorithms that can learn such models from implicit feedback in settings with resource constraints. Finally, there is a whole range of additional information that could be incorporated into the model. For example, a taxanomy (of words or documents) is likely to provide valuable features for modeling the dependencies between head document and the results in the second-level ranking.

\bibliographystyle{abbrv}

\bibliography{cikm2011}

\balancecolumns
% That's all folks!
\end{document}